\documentclass[11pt,a4paper,amsmath,amssymb,pdflatex,showpacs,pra,onecolumn,notitlepage,nopdfoutputerror]{quantumarticle}
\usepackage{dcolumn}
\usepackage{mathtools}
\usepackage{enumerate}
\usepackage{graphicx}
\usepackage{dcolumn}
\usepackage{bm}
\usepackage{physics}
\usepackage{comment}
\usepackage{booktabs}
\usepackage{color}
\usepackage{hyperref}
\allowdisplaybreaks

\usepackage{array}
\usepackage{dcolumn}
\usepackage{longtable}
\usepackage{booktabs}
\usepackage{enumerate}
\usepackage[inline]{enumitem}
\usepackage{bm}
\usepackage{xcolor}
\usepackage{comment}
\newcounter{one}
\newcommand{\expt}[1]{\langle #1 \rangle}

\newcommand{\ra}{\rightarrow}
\newcommand{\eps}{\epsilon}

\newcommand{\nd}{\textendash}

\newtheorem{theorem}{Theorem}

\newtheorem{lemma}{Lemma}

\newcommand{\eq}[1]{\begin{align} #1 \end{align}}

\newcommand{\subeq}[1]{\begin{subequations} #1 \end{subequations}}

\def\QED{\mbox{\rule[0pt]{1.5ex}{1.5ex}}}

\def\endproof{\hspace*{\fill}~\QED\par\endtrivlist\unskip}

\newcommand{\be}{\beta}

\newcommand{\aeq}{&=}

\newcommand{\aeqd}{&:=}

\newcommand{\aeqle}{& \le}
\newcommand{\aeqge}{& \ge}

\newcommand{\mQ}{\mathcal{Q}}
\newcommand{\mR}{\mathcal{R}}

\newcommand{\mC}{\mathcal{C}}
\newcommand{\mD}{\mathcal{D}}
\newcommand{\mL}{\mathcal{L}}
\newcommand{\mT}{\mathcal{T}}
\newcommand{\mU}{\mathcal{U}}
\newcommand{\mV}{\mathcal{V}}

\newcommand{\mN}{\mathcal{N}}
\newcommand{\mE}{\mathcal{E}}
\newcommand{\mF}{\mathcal{F}}

\newcommand{\dg}{^\dagger}
\newcommand{\tl}{\tilde}
\newcommand{\defe}{:=}
\newcommand{\ep}{\varepsilon}

\newcommand{\Ga}{\Gamma}
\newcommand{\al}{\alpha}
\newcommand{\bu}{\bullet}
\newcommand{\sig}{\sigma}
\newcommand{\f}{\frac}
\newcommand{\half}{\frac{1}{2}}
\newcommand{\pr}{\prime}
\newcommand{\dl}{\delta}
\newcommand{\Dl}{\Delta}
\newcommand{\lm}{\lambda}
\newcommand{\Lm}{\Lambda}
\newcommand{\any}{^{\forall}}
\newcommand{\om}{\omega}
\newcommand{\Om}{\Omega}
\newcommand{\tot}{{\rm{tot}}}
\newcommand{\com}{, \quad}
\newcommand{\spa}{\quad}
\newcommand{\la}{\label}
\newcommand{\ke}[1]{ \vert #1 \rangle }
\newcommand{\no}{\nonumber}

\newcommand{\re}[1]{(\ref{#1})}
\newcommand{\res}[1]{\S \ref{#1}}

\newcommand{\hs}{\hspace}

\newcommand{\RM}[1]{{\rm{#1}}}
\newcommand{\p}{\partial}

\newcommand{\nor}[1]{\vert\vert #1 \vert\vert}
\newcommand{\bea}[1]{\begin{align}
#1
\end{align}}

\newcommand{\calA}{\mathcal{A}}

\newcommand{\calE}{\mathcal{E}}

\newcommand{\calI}{\mathcal{I}}

\newcommand{\calL}{\mathcal{L}}
\newcommand{\calM}{\mathcal{M}}
\newcommand{\calN}{\mathcal{N}}

\newcommand{\calR}{\mathcal{R}}

\newcommand{\calT}{\mathcal{T}}
\newcommand{\calU}{\mathcal{U}}

\newcommand{\calY}{\mathcal{Y}}

\begin{document}
\title{Speed-Accuracy Trade-Off Relations in Quantum Measurements and Computations}

\author{Satoshi Nakajima}
\affiliation{Graduate School of Informatics and Engineering, The University of Electro-Communications,1-5-1 Chofugaoka, Chofu, Tokyo 182-8585, Japan}
\email{satoshi.nakajima@uec.ac.jp}

\author{Hiroyasu Tajima}
\email{hiroyasu.tajima@uec.ac.jp}
\affiliation{Graduate School of Informatics and Engineering, The University of Electro-Communications,1-5-1 Chofugaoka, Chofu, Tokyo 182-8585, Japan}
\affiliation{JST, PRESTO, 4-1-8 Honcho, Kawaguchi, Saitama, 332-0012, Japan}

\begin{abstract}

In practical measurements, it is widely recognized that reducing the measurement time leads to decreased accuracy. 
However, whether an inherent speed-accuracy trade-off exists as a fundamental physical constraint for quantum measurements is not obvious, 
and the answer remains unknown. 
Here, we establish a fundamental speed-accuracy trade-off relation based on the energy conservation law and the locality. 
Our trade-off works as a no-go theorem that the zero-error measurement for the operators that are 
non-commutative with the Hamiltonian cannot be implemented with finite time. 
This relation universally applies to various existing errors and disturbances defined for quantum measurements. 
We furthermore apply our methods to quantum computations and provide another speed-accuracy trade-off relation for unitary gate implementations, 
which works as another no-go theorem that any error-less implementations of quantum computation gates changing energy cannot be implemented with finite time, 
and a speed-disturbance trade-off for general quantum operations.

\end{abstract}

\maketitle

\section{Introduction} \la{s_1}

When conducting measurements, a reduction in measurement duration frequently results in compromised accuracy. 
This tendency can be seen in several specific instances. In superconducting qubits, it is known that trying to shorten the measurement time 
(i.e. to increase the speed of the measurement) results in reduced accuracy 
\cite{PhysRevApplied.7.054020,PRXQuantum.5.010307}. 
In electrochemical impedance measurements, it is also known that if the measurement time is short, the accuracy of the measurement drops \cite{GABRIELLI1982201}.

These examples, however, are specific to the nature of each measurement system, with distinct underlying mechanisms.
Whether a speed-accuracy trade-off due to fundamental physical laws exists in general quantum measurements is 
a non-trivial question, and the answer has remained unknown. 
It is important to distinguish that the speed-accuracy trade-off discussed here is different from the time-energy uncertainty relation
\cite{Bohr28, Heisenberg27, Landau31, Mandelstam45, Aharonov61, Ban93,  Busch08, Sagawa16}, 
which addresses the trade-off between the precision of time measurements and energy fluctuations, without restricting the relationship between measurement time and accuracy.
On the other hand, if a speed-accuracy trade-off relation exists for quantum measurements, 
it requires that a measurement of a certain desired accuracy can never be realized in a smaller measurement time than some threshold.

In this paper, we establish a speed-accuracy trade-off relation in quantum measurements as a consequence of the energy conservation law and the local interaction in the measurement device.  Assuming the validity of conservation law and the Lieb-Robinson bound \cite{Lieb72, Nachtergaele06}, a consequence of the locality property in quantum many-body systems, we derive a universal trade-off relation between measurement time and measurement error when measuring an observable that does not commute with the Hamiltonian of the system of interest. Our results work as a no-go theorem for quantum measurements under the energy conservation law and the locality: \textit{any error-less measurement of a physical quantity that is non-commutative with the Hamiltonian requires an infinite measurement time}.

Our results universally hold against various errors and disturbances that have been previously defined for quantum measurements. 
Reflecting the long history of quantum measurement, there are various definitions of errors and disturbances, e.g. Arthurs\nd Kelly\nd Goodman (AKG) \cite{Arthurs65,Yamamoto86,Arthurs88,Ishikawa91,Ozawa91}, 
Ozawa \cite{Ozawa03,Ozawa04,Ozawa19,Ozawa21}, 
Watanabe\nd Sagawa\nd Ueda (WSU) \cite{Watanabe11PRA,Watanabe11ARX}, 
Lee\nd Tsutsui (LT) \cite{Lee20ARX,Lee20ENT,Lee22}, 
Busch\nd Lahti\nd Werner (BLW) \cite{Busch13,Busch14PRA,Busch14RMP}, 
and the gate fidelity error of measurement channels, each with different operational meanings.
To address these errors and disturbances concurrently, we employ the irreversibility-based error and disturbance \cite{ET2023} 
and derive our speed-accuracy trade-off relation for them. 
Since the irreversibility-based error and disturbance recover or give lower bounds for the above existing errors and disturbances, 
our trade-off relation provides the speed-accuracy trade-off relations for all these errors and disturbances as corollaries.

Our methods apply not only to quantum measurements but also to general quantum operations. To illustrate this point, we establish a speed-accuracy trade-off for arbitrary unitary gates, which works as a no-go theorem for quantum computation gates (i.e. unitary gates) under the energy conservation law and the locality: \textit{any error-less implementation of unitary gate changing energy requires an infinite implementation time}.
We also give a trade-off between the time duration and the disturbance of a quantum operation, which is universally valid for a variety of existing disturbances. 

The structure of the paper is as follows. 
First, we explain the set up (\res{s_2}). 
Next, we introduce the irreversibility-based error and disturbance (\res{s_3}). 
In \res{s_4}, we explain the speed-accuracy trade-off relation for quantum measurements. 
In \res{s_5}, we explain the speed-accuracy trade-off relation for quantum computation gates. 
Next, we explain a sketch of proofs (\res{s_6}). 
In \res{s_7}, we summarize this paper. 
Appendix \ref{A_A}, we introduce techniques for the proofs. 
In Appendix \ref{A_B}, we prove the main theorems.

\section{Set up} \la{s_2}

Let us introduce our setup for quantum measurements.
As the system of interest $S$, we consider an arbitrary quantum system whose Hilbert space is a finite dimension. 
For this system $S$, we implement a general measurement process $\calE$:
\eq{
\calE(\rho):=\sum_j\calE_j(\rho)\otimes\ket{j}\bra{j}_P .\label{measurement1}
}
Here, $P$ is a memory system, and $\{\calE_j\}$ are completely positive (CP) maps from $S$ to $S'$ such that $\sum_j\calE_j$ is a completely positive and trace preserving (CPTP) map.

To explore the relationship between measurement time and accuracy, we consider the above measurement $\calE$ to be realized over a finite time $t$ by coupling the system $S$ to another quantum system $E$, representing the measurement apparatus. 
We suppose $E$ and $SE$ include the memory system $P$ and $S'P$, respectively. 
We also suppose that $E$ is a locally interacting system where the Lieb-Robinson bound \cite{Lieb72, Nachtergaele06} holds. 
In other words, we consider the device $E$ as a generic lattice system that consists of many subsystems called \textit{sites} whose Hilbert spaces' dimensions can be different from each other and assume the local-interaction conditions presented below to hold in $E$. The Hamiltonian $H_E$ of $E$ has the form 
$H_E = \sum_{Z \subset \Lm_E }h_Z$
where $h_Z$ is the local  Hamiltonian of the compact support $Z$. 
The set $\Lm_E$ is the set of all sites in $E$. 
The Hamiltonian of the total system is given by
$H_\tot = H_S+H_\RM{int} + H_E$.
Here, $H_S$ is the system Hamiltonian and $H_\RM{int}$ is the interaction Hamiltonian. 
Without loss of generality, we consider $S$ as one site (denoted by $0$). 
Then the total Hamiltonian $H_\tot$ can be rewritten as
\bea{
H_\tot \aeq \sum_{Z \subset \Lm_E \cup \{0\}} h_Z \no\\
\aeq h_{\{0\}}+\Big(\sum_{Z \ni 0}h_Z-h_{\{0\}} \Big) +\sum_{Z \subset \Lm_E} h_Z. \la{eq_Hamiltonian}
}
The first, second, and third terms correspond to $H_S$, $H_\RM{int}$, and $H_E$ respectively. 
Due to the assumption that the Lieb-Robinson bound holds in the system $E$, the Hamiltonian $H_E$ satisfies the local-interaction condition such that there exist $\lm>0$, $\mu>0$, and $p_0>0$ such that for any $x, y \in \Lm_E \cup \{0\}$, 
\subeq{
\eq{
 \sum_{Z \ni x, y} \nor{h_Z} \aeqle \lm e^{-\mu d(x,y)} ,\la{katei_LR1}\\
\sum_z e^{-\mu [d(x,z)+d(y,z)]} \aeqle p_0 e^{-\mu d(x,y)}. \la{katei_LR2}
}
}
Here, $d(x,y)$ is the distance defined by the shortest path length from the site $x$ to the site $y$, and $\nor{X}$ is the operator norm of $X$. 
Under the above assumption, the Lieb-Robinson bound holds for arbitrary Hermitian operators $A$ and $B$ whose supports are included by two regions $X\subset \Lm_E \cup \{0\}$ and $Y\subset \Lm_E \cup \{0\}$, respectively \cite{Lieb72, Nachtergaele06, Iyoda17}:
\bea{
\nor{[A(t),B]} \le  C \nor{A} \cdot \nor{B} \cdot \abs{X}\cdot \abs{Y}e^{-\mu d(X,Y)}  (e^{v t}-1) ,\la{e_LRB}
}
where $C$ is a constant $C \defe 2/p_0$, $v$ is the Lieb-Robinson velocity $v \defe 2\lm p_0$, and $d(X,Y)$ is the distance between $X$ and $Y$, 
defined as $d(X,Y) \defe \min_{x \in X, y \in Y}d(x,y)>0$.
Here, $A(t):=e^{iH_{\tot}t}Ae^{-iH_{\tot}t}$ and $\abs{X}$ is the number of the elements of $X$. 

Using the above measurement device $E$ and the initial state $\rho_\tot(0) =\rho \otimes \rho_E$, we realize the measurement $\calE$ as follows:
\bea{
\calE(\rho) \aeqd \tr_{E^\pr}[e^{-iH_\tot t}\rho_\tot(0) e^{iH_\tot t}] .\label{SS_measurement}
}
Here, $E'$ is the subsystem remaining after removing $S'P$ from the entire system $SE$.

In this paper, we also treat the speed-accuracy trade-off for quantum computation gates and the trade-off between disturbance and operation time for general quantum operations.
To treat such cases, we only have to remove the memory system $P$ from the above setup.
Then, for a general CPTP map $\Lambda$ from $S$ to $S'$, we obtain
\eq{
\Lambda(\rho) \aeqd \tr_{E^\pr}[e^{-iH_\tot t}\rho_\tot(0) e^{iH_\tot t}] .\label{SS_channel}
}

\section{Error and Disturbance} \la{s_3}

This paper aims to clarify the speed-accuracy trade-off relations for quantum measurements. Therefore, we need to quantify the error and disturbance of quantum measurements.
Although various definitions of error and disturbance exist, we can unify them using the concept of irreversibility.
Since lower bounds for irreversibility-based error and disturbance invariably apply to other existing errors and disturbances, such as
Arthurs\nd Kelly\nd Goodman (AKG) \cite{Arthurs65,Yamamoto86,Arthurs88,Ishikawa91,Ozawa91}, 
Ozawa \cite{Ozawa03,Ozawa04,Ozawa19,Ozawa21}, 
Watanabe\nd Sagawa\nd Ueda (WSU) \cite{Watanabe11PRA,Watanabe11ARX}, 
Lee\nd Tsutsui (LT) \cite{Lee20ARX,Lee20ENT,Lee22}, 
and Busch\nd Lahti\nd Werner (BLW) \cite{Busch13,Busch14PRA,Busch14RMP}, 
we employ the irreversibility-based error and disturbance and derive trade-off relations for them and the measurement time.

To introduce the irreversibility-based error and disturbance, we first introduce the concept of the irreversibility of quantum channels. Let us consider a CPTP map $\calL$ from a system $K$ to another system $K^\pr$ 
and an arbitrary test ensemble $\Omega=\{p_{k},\rho_{k}\}$. 
Here, $\{\rho_{k}\}$ is a set of quantum states in $K $ and $\{p_{k}\}$ are its preparation probabilities.
We define the irreversibility of $\calL$ as follows \cite{Tajima22}:
\bea{
\delta(\calL,\Omega):=\min_{\calR}\sqrt{\sum_{k}p_{k}\delta^{2}_{k}}.
}
Here, $\delta_{k}:=D_{F}(\rho_{k},\calR\circ\calL(\rho_{k}))$ and $\calR$ runs over CPTP maps (recovery map) from $K^\pr$ to $K$, 
$D_{F}(\rho,\sigma):=\sqrt{1-F(\rho,\sigma)^{2}}$ is the purified distance, 
and $F(\rho,\sigma):=\tr \sqrt{\sqrt{\rho}\sigma\sqrt{\rho}}$ is the Uhlmann fidelity.

Using the above irreversibility measure, we introduce the irreversibility-based error \cite{ET2023} which can unify various existing errors of quantum measurements. 
The core idea is to convert the measurement $\calE$ to another quantum channel.
Let $Q$ be a qubit system and $\calM$ be a CPTP map from the system $S$ to $S'P$. 
We introduce a CPTP map from $Q$ to $P Q$ as
\bea{
\calL_{\rho,A,\theta,\calM}&:=\calT_{S'}\circ\calM\circ\calU_{A,\theta}\circ\calA_{\rho}.
}
Here, $\calA_{\rho}(...):=\rho \otimes (...)$ where $\rho$ is a quantum state on $S$, 
$\calU_{A,\theta}(...):=e^{-i\theta A\otimes\sigma_{z}}(...)
e^{i\theta A\otimes\sigma_{z}}$ with the Pauli-$z$ operator $\sigma_{z}$ in $Q$ and an observable $A$ on $S$, and $\calT_{S'}$ is the partial trace of $S'$, i.e. $\calT_{S'}(...):=\Tr_{S'}[...]$.
Using this channel, we define the irreversibility-based disturbance as 
\bea{
\eps(\rho,A,\calE)&:=\lim_{\theta\ra+0}\frac{\delta(\calL_{\rho,A,\theta,\calE},\Omega_{0})}{\theta}\label{ET_error}
}
where $\Omega_{0}$ is a test ensemble $\{(1/2, \ket{+}\bra{+}_Q), (1/2, \ket{-}\bra{-}_Q)\}$, $\ket{\pm}$ are eigenvectors of the Pauli-$x$ operator $\sigma_{x}$ in $Q$. 

Since the irreversibility-based error $\eps(\rho,A,\calE)$ bounds existing errors of quantum measurements from lower \cite{ET2023}, lower bounds for $\eps(\rho,A,\calE)$ are always applicable to other errors such as the errors defined by AKG, Ozawa, WSU, LT, and BLW.
Indeed, all of the results for $\eps(\rho,A,\calE)$ in this paper are applicable to these errors.

Next, we introduce the irreversibility-based disturbance \cite{ET2023}. 
To introduce the disturbance, we introduce another CPTP map from $Q$ to $S' Q$ as
\bea{
\calL'_{\rho,A,\theta,\calM}&:=\calT_{P}\circ\calM\circ\calU_{A,\theta}\circ\calA_{\rho},
}
where $\calT_{P}$ is the partial trace of $P$.
Using this channel, we define the irreversibility-based disturbance as
\bea{
\eta(\rho,A,\calE)&:=\lim_{\theta\ra+0}\frac{\delta(\calL'_{\rho,A,\theta,\calE},\Omega_{0})}{\theta} . \label{def_dis}
}
Again, since the irreversibility-based disturbance $\eta(\rho,A,\calE)$ bounds existing disturbance of quantum measurements from lower \cite{ET2023}, lower bounds for $\eta(\rho,A,\calE)$ are always applicable to other existing disturbances.

While the error can be defined only for the measurement, the disturbance can be defined for a general CPTP map that is not necessarily a measurement.
For such a general CPTP map $\Lambda$ from $S$ to $S'$, we can also define $\eta(\rho,A,\Lambda)$ by substituting $\calL''_{\rho,A,\theta,\Lambda}:=\Lambda\circ\calU_{A,\theta}\circ\calA_{\rho}$ for $\calL'_{\rho,A,\theta,\calE}$ in \eqref{def_dis}.

\section{Result 1. Speed-accuracy trade-off relation for quantum measurements} \la{s_4}

Now, let us discuss the limitation on the quantum measurement imposed by local interaction.
We first demonstrate that when the measurement device $E$ satisfies the local interaction conditions \eqref{katei_LR1} and \eqref{katei_LR2}, 
i.e. the Lieb-Robinson bound holds in the measurement device, 
a universal trade-off relation between the measurement time and the error of the measurement is always present.  
\begin{theorem} \la{theo_main_e}
Let $\calE$ be a quantum measurement \eqref{measurement1} which can be implemented with the measurement time $t$.
We suppose the Yanase condition $[H_P,\ket{j}\bra{j}_P]=0$ for any $j$, where $H_P$ is the Hamiltonian of the memory system. 
Then, the following relation holds: 
\bea{
\ep(\rho, A, \mE)
\aeqge \half \f{\abs{\expt{[H_S, A]}_\rho}}{K(R_0+R_1t)^d+c_1}.  \la{e_main_e}
}
Here, $\expt{\bu}_\sig\defe \tr(\sig \bu)$, $c_1$, $K>0$, and $R_0>0$ are constants. 
$R_1$ is a positive constant satisfying $-\mu l_0R_1+v<0$, where $l_0$ is the unit of the length. 
And $d$ is the space dimension of $E$. 
\end{theorem}

Theorem \ref{theo_main_e} acts as a no-go theorem, indicating that any error-less measurement of a physical quantity 
that is non-commutative with the Hamiltonian requires an infinite measurement time.
As noted in the previous section, since $\ep(\rho, A, \mE)$ is a lower bound for the existing errors defined by AKG, Ozawa, WSU, LT, and BLW \cite{ET2023}, 
these errors can be substituted for $\ep(\rho, A, \mE)$ in the inequality \eqref{e_main_e}. 
Furthermore, this inequality can be extended to cases where the Yanase condition does not hold, 
i.e. when $[H_P,\ket{j}\bra{j}_P] \ne 0$, by replacing $H_S$ in \re{e_main_e} with $Y_S \defe H_S - \mE\dg\circ\calT_{S'}\dg (H_P) $. 
Here, for a CPTP map $\Gamma(\bu)$, $\Gamma^\dagger(\bu)$ is adjoint of $\Gamma$ 
satisfying $\expt{\Gamma^\dagger(O)}_{\xi}=\expt{O}_{\Gamma(\xi)}$ for any operator $O$ and any state $\xi$. 

A similar theorem for the disturbance of quantum measurements also holds.
\begin{theorem} \la{theo_main_d}
Let $\calE$ be a quantum measurement \eqref{measurement1} that can be implemented with the measurement time $t$.
Then, the following relation holds: 
\bea{
\eta(\rho, A, \mE) \ge \half \f{\abs{\expt{[Y'_S,A]}_\rho}}{K(R_0+R_1t)^d+c} . \la{e_main_d}
}
Here,  $Y'_S \defe H_S-\mE\dg\circ\calT_{P}\dg (H_{S^\pr})$. 
$c$ is a constant. 
$K$, $R_0$, $R_1$, and $d$ are the same as those in Theorem \ref{theo_main_e}. 
\end{theorem}
Theorem \ref{theo_main_d} means that a disturbance-free quantum measurement requires an infinite measurement time. 
Furthermore, the inequality \eqref{e_main_d} can be extended to an arbitrary CPTP map $\Lambda$ 
by substituting $\eta(\rho, A, \Lambda)$ and $Y''_S \defe H_S-\Lambda\dg (H_{S^\pr})$ for $\eta(\rho, A, \mE)$ and $Y'_S$, respectively.

\section{Result 2. Speed-accuracy trade-off relation for quantum computation gates} \la{s_5}

The speed-accuracy trade-off relation is not limited to quantum measurements; it also applies to unitary gates. 
In other words, when implementing a quantum computation gate with local interactions, 
a universal trade-off relation between the operation time and the implementation error of the gate always holds.
\begin{theorem} \la{theo_gfe}
Let $\mU$ be a unitary map on $S$.
Suppose $\mU$ can be approximately implemented by a CPTP map $\Lambda$ that can be realized as \eqref{SS_channel} with the operation time $t$ ($S'=S$ case in \eqref{SS_channel}).
Then, between the implementation error $D_F(\Lambda, \mU) \defe \max_\rho D_F(\Lambda(\rho), \mU(\rho))$ and the operation time $t$, 
the following trade-off inequality holds:
\bea{
D_F(\Lambda, \mU) \ge \f{1}{4} \f{C_U}{K(R_0^\pr+R_1 t)^d+K_1}. \la{gfe_goal}
}
Here, $C_U$ is given by 
\bea{
C_U \aeqd \max_{\ke{\psi_+}, \ke{\psi_-}, \langle \psi_+ \ke{\psi_-}=0} \abs{\bra{\psi_+}[H_S-\mU\dg(H_S)]\ke{\psi_-}} .\la{def_C_U}
}
$K$, $R_1$, and $d$ are the same constants as in Theorem \ref{theo_main_e}. 
$R_0^\pr> 0$ and $K_1$ are constants.
\end{theorem} 
Theorem \ref{theo_gfe} serves as a no-go theorem for quantum computation gates.
The quantity $C_U$ is greater than zero if and only if $H_S-\calU^\dagger(H_S)\ne0$, i.e. the unitary $\calU$ changes the energy.
Thus, Theorem \ref{theo_gfe} asserts that any energy-changing unitary cannot be implemented with zero error in finite operation time.

\section{Sketch of proof} \la{s_6}

Here, we present a brief overview of the proofs for our main results. 
Detailed proofs are provided in the Appendix \ref{A_A} and Appendix \ref{A_B}.

Our main results are based on three observations.
The first observation is that when we perform a quantum operation $\Lambda$ on the target system $S$ as a result of the finite-time interaction with $E$, 
the amount of coherence in $E$ that we can utilize as a quantum resource to realize the operation is restricted by the operation time (Fig.\ref{strategy}).
This limitation arises from the local interaction within $E$. 
Due to the Lieb-Robinson bound, only the part of $E$ in proximity to $S$ actually impacts the implementation within a finite time $t$.
By extending the techniques in Refs.\cite{Iyoda17,Shiraishi_Tajima}, 
we can divide $E$ into $E_1$ (near $S$) and $E_2$ (far from $S$), showing that $E_2$ has negligible influence on $\Lambda$.
Consequently, the coherence within $E_1$, which depends on $t$, is the only resource available for implementing $\Lambda$. 

\begin{figure}[tb]
		\centering
		\includegraphics[width=.75\textwidth]{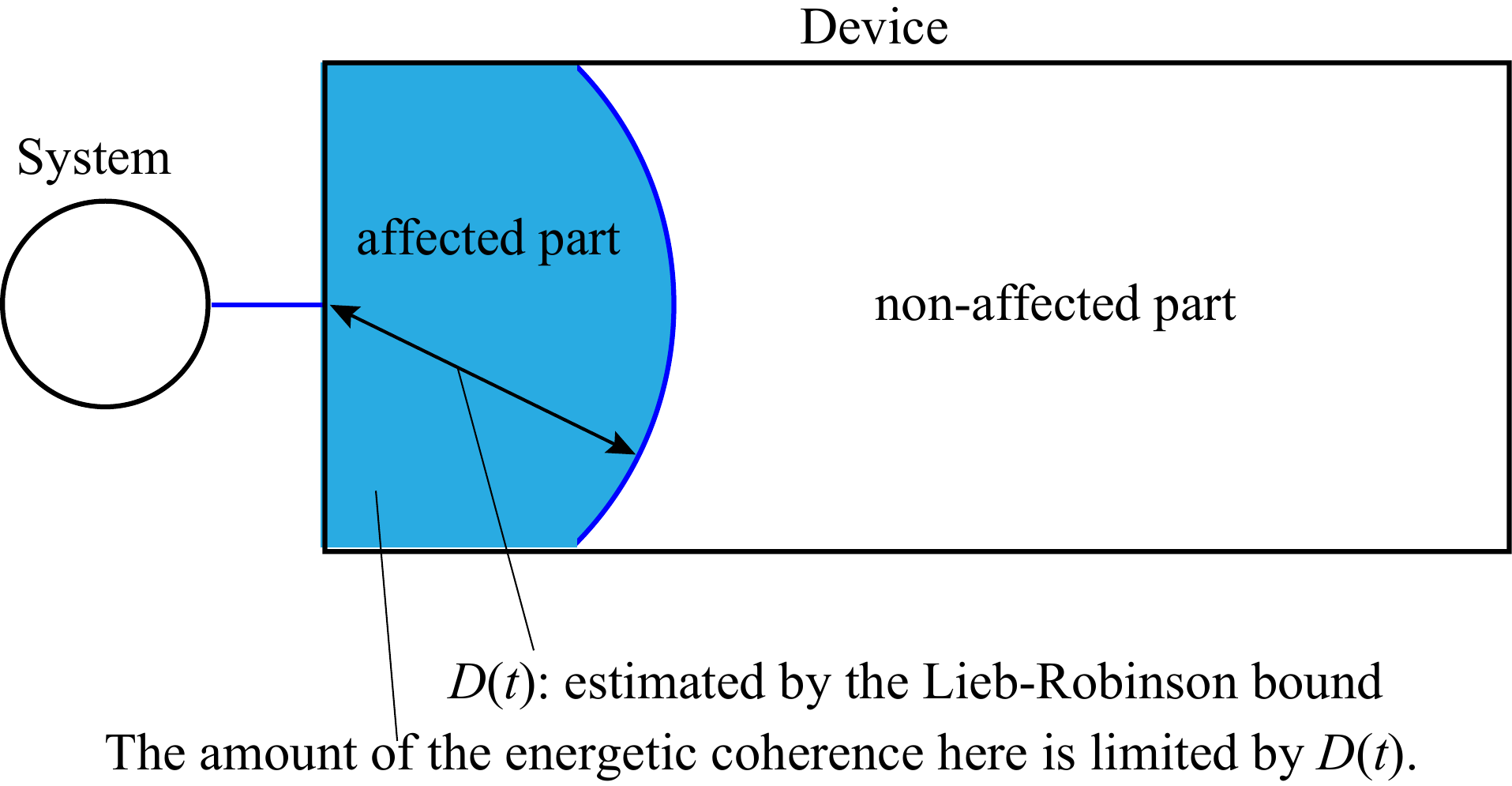}
		\caption{Schematic illustration of the proof strategy. 
		Due to the local interactions within the device, 
		we can approximately separate the device into the part affected by the system in the operation time $t$ and the other non-affected part. 
		Using the Lieb-Robinson bound, we can estimate the size of the affected part as the function of $t$. 
		The amount of energetic coherence in the affected part is limited by the size of the affected part. 
		Therefore, the amount of energetic coherence that can be utilized to implement the desired channel $\Lambda$ is limited by $t$ (the first observation). 
		The irreversibility of the channel $\Lambda$ is approximately inversely proportional to the amount of the energetic coherence used to implement $\Lambda$ 
		(the second observation). 
		When the desired operation (i.e., the operation that $\Lambda$ approximates) is a measurement or a unitary gate, 
		we can convert the irreversibility of the channel $\Lambda$ to the error of the channel (the third observation). 
		By combining these three observations, we derive the speed-accuracy trade-off relations for measurements and computations.}
		\label{strategy}
\end{figure}

The second observation is that if $\Lambda$ locally changes the energy and is approximately reversible, 
its implementation necessitates significant energetic coherence. 
This requirement stems from the energy conservation law. 
When the energy conservation law holds in the whole system, implementing a channel that changes energy locally 
demands energetic coherence inversely proportional to the channel's irreversibility, $\delta(\Lambda,\Omega)$.
To obtain this observation in a usable form in the present setting, we utilize the resource theory of asymmetry \cite{Gour2008resource,Marvian_2013,Marvian_thesis,skew_resource,Takagi_skew,Marvian_distillation,YT,YT2,Kudo_Tajima,Shitara_Tajima}, 
which treats conservation laws and symmetries in quantum systems. 
To be concrete, we extend a trade-off relation between symmetry, irreversibility, and energetic coherence \cite{Tajima22}, 
which is a unification of restrictions on quantum information processing by symmetry, e.g. the Wigner-Araki-Yanase theorem 
\cite{Wigner1952,Araki-Yanase1960,OzawaWAY,Korzekwa_thesis,TN,ET2023} for measurements, 
the coherence-error tradeoff for unitary gates \cite{ozawaWAY_CNOT,Karasawa_2009,TSS,TSS2,TS}, 
 the Eastin-Knill theorems \cite{Eastin-Knill,e-EKFaist,e-EKKubica,TS} for error correcting codes, and the coherence-error tradeoff for Gibbs-preserving maps \cite{Tajima_Takagi}, to a usable form in this situation.
Combining the above two observations, we can see that when any quantum channel is implemented with finite operation time, 
the irreversibility of the channel is bounded from below by the operation time. 

To obtain the speed-accuracy trade-off relations from the above observations, 
we use the third observation: the errors of the quantum measurement and the quantum computation can be interpreted as the special cases of irreversibility.
For the quantum measurement error, we can use the irreversibility-based error \eqref{ET_error}.
Using the irreversibility-based error, we can convert the discussion of irreversibility to the discussion of the measurement error, 
and obtain the speed-accuracy trade-off for quantum measurements, i.e. Theorem \ref{theo_main_e}. 
Similarly, using \eqref{def_dis}, we derive Theorem \ref{theo_main_d}.
As the irreversibility-based error and disturbance are lower bounds for other errors and disturbances, 
any established error or disturbance can be substituted in these relations. 
For quantum computations, we establish the bound for implementation error $D_F(\Lambda,\calU)$:
\bea{
D_F(\Lambda,\calU) \ge \dl(\Lambda, \Om). \la{l_gfe_1_main}
}
This inequality holds for an arbitrary test ensemble $\Om$.
Because of these relations, again we can covert the discussion of irreversibility to the discussion of the implementation error of quantum computation gates, 
and obtain the speed-accuracy trade-off for quantum computations (Theorem \ref{theo_gfe}).

\section{Summary} \la{s_7}

We established speed-accuracy trade-off relations in quantum measurements and computations. 
Our main results work as no-go theorems meaning that quantum measurements and operations with infinite precision take an infinite time. 
To derive the main results, we give the speed-irreversibility trade-off for an arbitrary CPTP map, 
and convert it to the speed-accuracy trade-off with using the fact that the error of the measurement and the unitary gate 
are the special cases of the irreversibility $\delta$.

Since our results are consequences of energy conservation law and locality of interaction, 
these results are valid even if the time evolution of the total system is complex, e.g. non-Markov dynamics.
When we restrict the dynamics of the system to Markov dynamics, in particular a quantum master equation, 
we might give another type of trade-off relation between the measurement accuracy and the measurement time, but we leave this question as a future work.

\acknowledgments
This work was supported by MEXT KAKENHI Grant-in-Aid for Transformative
Research Areas B ``Quantum Energy Innovation” Grant Numbers JP24H00830 and JP24H00831.
HT was supported by JST PRESTO No. JPMJPR2014, JST MOONSHOT No. JPMJMS2061. 
In this paper, the authors obtained the help of ChatGPT 4o and DeepL Translator to correct English expressions naturally.
These helped only in the final stages of revising the text in the manuscript; the authors alone derived the research ideas and results, and wrote the original manuscript.

\appendix

\section{Preliminary: Introduction of techniques} \la{A_A}

\subsection{SIQ trade off relation}

Let us consider a completely positive trace-preserving (CPTP) map $\calL$ from a system $K$ to another system $K^\pr$ 
and an arbitrary test ensemble $\Omega=\{(p_{k},\rho_{k})\}$. 
Here, $\{\rho_{k}\}$ is a set of quantum states in $K $ and $\{p_{k}\}$ are its preparation probabilities.
We define the irreversibility of $\calL$ as follows \cite{Tajima22}:
\bea{
\delta(\calL,\calR,\Omega):=\sqrt{\sum_{k}p_{k}\delta^{2}_{k}} \com
\delta_{k}:=D_{F}(\rho_{k},\calR\circ\calL(\rho_{k})).
}
Here, $\calR$ is a CPTP map (recovery map) from $K^\pr$ to $K$,  $D_{F}(\rho,\sigma):=\sqrt{1-F(\rho,\sigma)^{2}}$ is the purified distance, 
and $F(\rho,\sigma):=\tr \sqrt{\sqrt{\rho}\sigma\sqrt{\rho}}$ is the fidelity.

Let $\mF_\sig(W)$ be the symmetric-logarithmic-derivative (SLD) Fisher information for
the state family  $\{e^{-iW\theta}\sig e^{iW\theta}\}$. 
$\mF_\sig(W)$ is given by
\bea{
\mF_\sig(W) \aeqd \expt{L_W, L_W }_\sig^\RM{SLD} 
}
where $\expt{O_1, O_2}_\sig^\RM{SLD}\defe \tr[\sig(O_1O_2+O_2O_1)/2] $.
The SLD $L_W$ is defined as $-i[W,\sig]=(L_W\sig+\sig L_W)/2$. 
For an arbitrary CPTP map $\calI$ from $S_1$ to $S_2$, 
$\calI^{\dagger}$ is defined by 
$\tr_{S_2}[Y_2\calI(Y_1)]=\tr_{S_1}[Y_1\calI^{\dagger}(Y_2)]$.

In Reference \cite{Tajima22}, a universal trade-off structure between symmetry, irreversibility, and quantum
coherence (SIQ) has been studied. 
The following the SIQ theorem holds \cite{Tajima22}:

\begin{theorem} \la{theo_0}

We consider a CPTP map $\mL$ from $K$ to $K^\pr$ given by $\mL(\bu)=\tr_{R^\pr}(U \bu \otimes \rho_R U\dg)$
and suppose that $U\dg(X_{K^\pr}+X_{R^\pr})U-(X_K+X_R)=0$ holds. 
Here, $X_\bu$ $(\bu=K,R,K^\pr,R^\pr)$ is a Hermitian operator on the system $\bu$ and $U$ is an unitary operator on $KR=K^\pr R^\pr$.  
We suppose that the Hilbert spaces of $R$ and $R^\pr$ are finite dimensions. 
Then, for a test ensemble $(\{(p_k,\rho_k)\})$ that $\{ \rho_k \}$ are orthogonal to each other, the relation
\bea{
\dl(\mL, \Om) \sqrt{1-\min_k p_k} \ge \f{\mC}{\Dl+\sqrt{\mF}} \la{SIQ_theorem_0}
}
holds. Here,
\bea{
\dl(\mL, \Om) \defe \min_\mR \dl(\mL, \mR, \Om).
}
$\mF \defe \mF_{\rho_R}(X_R)$ and 
\bea{
\Dl \aeqd \max_{\rho \in \cup_k \RM{supp}(\rho_k)}\sqrt{\mF_{\rho\otimes \rho_R}(X_K-U\dg X_{K^\pr}U)} 
}
where the maximum runs over the subspace spanned by the supports of the test states $\{ \rho_k \}$.
$\mC$ is defined by
\bea{
\mC \aeqd \sqrt{\sum_{k \ne k^\pr}p_k p_{k^\pr} \tr_K[\rho_k Y_K \rho_{k^\pr} Y_K] } 
}
with $Y_K \defe X_K-\mL\dg (X_{K^\pr})$.

\end{theorem}

\subsection{Error and disturbance} 

We introduce the general disturbance \cite{ET2023}. 
Let $Q$ be a qubit system and $\calN$ be a CPTP map from the system $\al$ to $\al^\pr$. 
We introduce a CPTP map from $Q$ to $\al^\pr Q$ as
\bea{
\label{eq:loss_proc_error}
\calL_{\rho,A,\theta,\calN}&:=\calN\circ\calU_{A,\theta}\circ\calA_{\rho}.
}
Here, $\calA_{\rho}(...):=\rho \otimes (...)$ where $\rho$ is quantum state on $\al$, 
$\calU_{A,\theta}(...):=e^{-i\theta A\otimes\sigma_{z}}(...)
e^{i\theta A\otimes\sigma_{z}}$ with the Pauli-$z$ operator $\sigma_{z}$ in $Q$ and a system observable $A$. 
We define the disturbance as
\bea{
\eta(\rho,A,\calN,\calR)&:=\lim_{\theta\ra+0}\frac{\delta(\calL_{\rho,A,\theta,\calN},\calR,\Omega_{0})}{\theta}
}
where $\calR$ is a recovery map from $\al^\pr Q$ to $Q$, 
$\Omega_{0}$ is a test ensemble $\{(1/2, \ket{+}\bra{+}_Q), (1/2, \ket{-}\bra{-}_Q)\}$, $\ket{\pm}$ are eigenvectors of the Pauli-$x$ operator $\sigma_{x}$ in $Q$. 

Next, we introduce the general error \cite{ET2023}. 
Let $\calM$ be a CPTP map from the system $\al$ to a memory system $P$. 
We define the error as 
\bea{
\eps(\rho,A,\calM,\calR)&:=\lim_{\theta\ra+0}\frac{\delta(\calL_{\rho,A,\theta,\calM},\calR,\Omega_{0})}{\theta}
}
where $\calR$ is a recovery map from $P Q$ to $Q$.

We consider a  CPTP map $\mD$ from $\al$ to $\al^\pr P$. 
The general error and disturbance of $\mD$ are defined as 
\bea{
\ep(\rho,A,\mD,\mR) \aeqd  \ep(\rho,A,\mT_{\al^\pr} \circ \mD,\mR),\\
\eta(\rho,A,\mD,\mR) \aeqd  \eta(\rho,A,\mT_P \circ \mD,\mR).
}
Here, $\mT_\bu$ is the partial trace of $\bu$, i.e.  $\mT_\bu(X)=\tr_\bu(X)$. 
We also introduce 
\bea{
\ep(\rho,A,\mD) \aeqd \min_\mR \ep(\rho,A,\mD,\mR)  ,\\
\eta(\rho,A,\mD) \aeqd \min_\mR \eta(\rho,A,\mD,\mR) . 
}

\section{Formulation and main result} \la{A_B}

\subsection{Setting}

We consider an arbitrary quantum system $S$ of which Hilbert space is a finite dimension. 
This system couples to the measurement device $E$ of which Hilbert space is an infinite dimension.
We suppose $E$ includes a memory system $P$.
The device $E$ is a lattice system of which Hamiltonian $H_E$ has the form $H_E = \sum_{Z \subset \Lm_E }h_Z$
where $h_Z$ is the local  Hamiltonian of the compact support $Z$. 
The set $\Lm_E$ is the set of all sites in $E$. 
We suppose the initial state is 
\bea{
\rho_\tot(0) =\rho_S(0) \otimes \rho_E.
} 
The Hamiltonian of the total system is given by
\bea{
H_\tot \aeq H_S+H_\RM{int} + H_E.
}
$H_S$ is the system Hamiltonian and $H_\RM{int}$ is the interaction Hamiltonian. 
Without loss of generality, we consider $S$ as one site (denoted by $0$). 
$H_\tot$ can be rewritten as \re{eq_Hamiltonian}. 
We suppose the local-interaction condition \re{katei_LR1} and \re{katei_LR2} for any $x, y \in \Lm_E \cup \{0\}$.
We consider a subsystem $S^\pr$ of $SE$ and introduce a system $E^\pr$ as $S^\pr E^\pr=SE$.  
The Hamiltonian can be rewritten as
\bea{
H_\tot \aeq H_{S^\pr}+H_\RM{int}^\pr+H_{E^\pr} . 
}
We consider a subsystem $E_1^\pr$ of $E^\pr$, which is close to $S^\pr$. 
The Hamiltonian $H_{E^\pr}$ is decomposed as 
\bea{
H_{E^\pr} \aeq H_{E_1^\pr}+H_{\p E_1^\pr}+H_{E_2} 
}
where $E_1^\pr E_2=E^\pr$. 
We denote by $\p E_1^\pr$ the support of $H_{\p E_1^\pr}$. 
We suppose that (i) $E_1^\pr$ includes $I\backslash \{0\} $ and $I^\pr \backslash \{0\} $
where $I$ and $I^\pr$ are supports of $H_\RM{int}$ and $H_\RM{int}^\pr$ respectively, and
(ii) $d(S^\pr,\p E_1^\pr)>0$ and  $d(I^\pr,\p E_1^\pr)>0$. 
If $S^\pr \ne S$, we also suppose that $E_1^\pr$ includes $S$. 
We introduce
\bea{
H_T \aeqd H_{S^\pr}+H_\RM{int}^\pr+H_{E_1^\pr}.
}
Then, $H_\tot$ can be rewritten as
\bea{
H_\tot =H_T+H_{\p E_1^\pr}+H_{E_2}.
}
We also introduce 
\bea{
\rho_\tot(t) \aeqd e^{-iH_\tot t}\rho_\tot(0) e^{iH_\tot t} ,\\
\rho_T(t) \aeqd e^{-iH_T t}\rho_\tot(0) e^{iH_T t} ,\\
\Lm_{S^\pr}(\rho_S(0)) \aeqd \tr_{E^\pr}[\rho_\tot(t)] ,\\
\tl \Lm_{S^\pr}(\rho_S(0)) \aeqd \tr_{E^\pr}[\rho_T(t)].
}
$\tl \Lm_{S^\pr}$ can be rewritten as
\bea{
\tl \Lm_{S^\pr}(\bu) \aeq \tr_{E^\pr}[e^{iH_T t}\bu \otimes \rho_{E_1}e^{-iH_T t}]
}
with $\rho_{E_1} \defe \tr_{E_2} (\rho_E)$.

We take $E_1$ as $E_1E_2=E$ and introduce
\bea{
H_{E_1} \aeqd H_T - (H_S+H_\RM{int}) .
}
$H_{E_1}$ can be rewritten as
\bea{
H_{E_1} = \sum_{Z \subset \Lm_{E_1}\backslash \{0\}} h_Z. \la{H_E_1}
}
Here, $\Lm_{E_1}$ is the set of all sites in $E_1$ and $\Lm_{E_1}\backslash \{0\} \subset \Lm_E$. 

We also introduce 
\bea{
\mE(\rho_S(0)) \aeqd \tr_{E^{\pr\pr}}[\rho_\tot(t)] 
}
where $E^{\pr\pr}$ is the subsystem remaining after removing $S^\pr P$ from the entire system $SE$. 
The relation 
\bea{
\Lm_{S^\pr} = \mT_{P} \circ \mE
}
holds. We put
\bea{
\Lm_P \defe \mT_{S^\pr} \circ \mE.
}
Note that
\bea{
\ep(\rho, A, \mE, \mR) \aeq \ep(\rho, A, \Lm_P, \mR) \ge \ep(\rho, A, \mE)  ,\\
\eta(\rho, A, \mE, \mR) \aeq \eta(\rho, A, \Lm_{S^\pr}, \mR) \ge \eta(\rho, A, \mE).
}

\noindent{\textbf{Remark.}} For quantum computation gates, we only have to remove the memory system $P$ from the above setup.
In this case, we introduce
\bea{
\eta(\rho,A,\Lm_{S^\pr}) \aeqd \min_\mR \eta(\rho,A,\Lm_{S^\pr},\mR). 
}

\subsection{Strategy of derivation}

Because the Hilbert space of $E$ is an infinite dimension, we can not use Theorem \ref{theo_0} for $\Lm_{S^\pr}$.
We approximate $\Lm_{S^\pr}$ by $\tl \Lm_{S^\pr}$.  
First, we derive a lower bound of $\nor{\Lm_{S^\pr}(X) -\tl \Lm_{S^\pr}(X)}_1$ by using the Lieb-Robinson bound  (\res{s3}). 
Here, $\nor{X}_1\defe \tr\sqrt{X\dg X}$ is the trace norm of $X$. 
Using this lower bound and $\eta(\rho,A,\tl \Lm_{S^\pr},\calR)$, we derive a lower bound of  $\eta(\rho,A, \calE,\calR)$ (\res{s4}).
Because 
\bea{
Z_{SE_1} \defe e^{iH_Tt} (H_{S^\pr}+H_{E_1^\pr})e^{-iH_Tt}-( H_S+H_{E_1})\ne 0, \la{con_low_0}
}
we need to expand the SIQ theorem to cases where conservation laws are broken. 
We extend Emori-Tajima's framework \cite{ET2023} to this case (\res{s5} and \res{s6}).
Finally, we derive a lower bound of $\eta(\rho,A,\tl \Lm_{S^\pr},\calR)$ (\res{s7}).
Then, we obtain a lower bound of $\eta(\rho,A,\calE)$. 
Similarly, the error $\eps(\rho,A,\calE)$ is also lower bounded. 
In \res{s_Z}, we derive an upper bound of $\nor{Z_{SE_1}}$. 
In \res{s10}, we consider $S^\pr=S$ case.
In \res{s11}, we derive the main result. 
In \res{s_Yanase}, we consier the Yanase condition. 
In \res{s12}, we consier the Gate-fidelity error. 

\subsection{Approximation of dynamics} \la{s3}

\begin{lemma}

The relation
\bea{
\nor{\Lm_{S^\pr}(X) -\tl \Lm_{S^\pr}(X)}_1\le \nor{X}_1 T_\RM{LRB} \la{IKS}
}
holds. 
 $T_\RM{LRB}$ is defined as
\bea{
T_\RM{LRB} \defe \f{C}{v}\nor{H_{\p E_1^\pr}}\cdot \abs{S^\pr}\abs{\p E_1^\pr}e^{-\mu d(S^\pr, \p E_1^\pr)}(e^{vt}-vt -1). \la{def_T_LRB}
}
\end{lemma}

\begin{proof}
To prove \re{IKS}, we use a formula
\bea{
\nor{Y}_1 \aeq \max_{\nor{O}=1} \abs{\tr(OY)}. \la{F7}
}
For any operator $O_{S^\pr}$ with support $S^\pr$ satisfying $\nor{O_{S^\pr} }=1$, 
\bea{
&\hs{-10mm} \abs{\tr_{S^\pr}[O_{S^\pr}\{\Lm_{S^\pr}(X) -\tl \Lm_{S^\pr}(X)\}]} \no\\
 \aeq \abs{\tr_{S^\pr E^\pr} [(e^{iH_\tot t}O_{S^\pr}e^{-iH_\tot t}-e^{iH_T t}O_{S^\pr}e^{-iH_T t})X \otimes \rho_E]} \no\\
 \aeqle \nor{e^{iH_\tot t}O_{S^\pr}e^{-iH_\tot t}-e^{iH_T t}O_{S^\pr}e^{-iH_T t}} \cdot \nor{X \otimes \rho_E}_1 \no\\
 \aeq \nor{e^{iH_\tot t}O_{S^\pr}e^{-iH_\tot t}-e^{iH_T t}O_{S^\pr}e^{-iH_T t}} \cdot \nor{X}_1
} 
holds \cite{Iyoda17}. 
To evaluate the final line, we calculate it as
\bea{
&e^{iH_\tot t}O_{S^\pr}e^{-iH_\tot t}-e^{iH_T t}O_{S^\pr}e^{-iH_T t} \no\\
\aeq \int_0^t ds \ \f{d}{ds}(e^{iH_\tot s}e^{iH_T(t-s)}O_{S^\pr}e^{-iH_T(t-s)}e^{-iH_\tot s}) \no\\
\aeq \int_0^t ds \ \Big(e^{iH_\tot s}i(H_\tot-H_T)e^{iH_T(t-s)}O_{S^\pr}e^{-iH_T(t-s)}e^{-iH_\tot s} \no\\
&\spa -e^{iH_\tot s}e^{iH_T(t-s)}O_{S^\pr}e^{-iH_T(t-s)}i(H_\tot-H_T)e^{-iH_\tot s} \Big) \no\\
\aeq \int_0^t ds \ \Big(e^{iH_\tot s}i(H_{\p E_1^\pr}+H_{E_2})e^{iH_T(t-s)}O_{S^\pr}e^{-iH_T(t-s)}e^{-iH_\tot s} \no\\
&\spa -e^{iH_\tot s}e^{iH_T(t-s)}O_{S^\pr}e^{-iH_T(t-s)}i(H_{\p E_1^\pr}+H_{E_2})e^{-iH_\tot s} \Big) \no\\
\aeq -i\int_0^t ds \ e^{iH_\tot s}[e^{iH_T(t-s)}O_{S^\pr}e^{-iH_T(t-s)},H_{\p E_1^\pr}]e^{-iH_\tot s} .
}
Here, we used $[H_{E_2},O_{S^\pr}]=0$ and $[H_{E_2}, H_T]=0$. 
Then, we obtain
\bea{
\abs{\tr_{S^\pr}[O_{S^\pr}\{\Lm_{S^\pr}(X) -\tl \Lm_{S^\pr}(X)\}]} 
  \aeqle \nor{X}_1\int_0^t ds \ \nor{[e^{iH_T(t-s)}O_{S^\pr}e^{-iH_T(t-s)},H_{\p E_1^\pr}]} .
}
Using the Lieb-Robinson bound \re{e_LRB}, we obtain
\bea{
 &\hs{-10mm}\abs{\tr_{S^\pr}[O_{S^\pr}\{\Lm_{S^\pr}(X) -\tl \Lm_{S^\pr}(X)\}]}  \no\\
 \aeqle \nor{X}_1C\nor{O_{S^\pr}} \cdot \nor{H_{\p E_1^\pr}} \cdot \abs{S^\pr}\abs{\p E_1^\pr}e^{-\mu d(S^\pr,\p E_1^\pr)} 
 \int_0^t ds \ (e^{v(t-s)}-1) \no\\
\aeq \nor{X}_1C \nor{H_{\p E_1^\pr}} \cdot \abs{S^\pr}\abs{\p E_1^\pr}e^{-\mu d(S^\pr,\p E_1^\pr)} \f{e^{vt}-1-vt}{v} \no\\
\aeq \nor{X}_1 T_\RM{LRB}.
}
This inequality and \re{F7} lead to \re{IKS}. 
\end{proof}

\subsection{Approximation of disturbance and error} \la{s4}

\begin{lemma} \la{l_approx}

The relation
\bea{
\abs{\eta(\rho, A, \mE, \mR)^2-\eta(\rho, A, \tl \Lm_{S^\pr}, \mR)^2} 
\aeqle \Big(\nor{\mR_2}_1+2 \nor{\mR_1}_1\nor{A} +2 \nor{A}^2 \Big) T_\RM{LRB} \la{eta_goal_pre}
}
holds. 
$T_\RM{LRB}$ is defined by \re{def_T_LRB}. 
$\mR_1$ and $\mR_2$ are defined by
\bea{
\mR \aeq \mR_0 + \theta \mR_1 +  \theta^2 \mR_2 + O(\theta^3). \la{def_mR_k}
}
$\nor{\mR_k}_1$ $(k=1,2)$ are are defined by
\bea{
\nor{\mR_k}_1 \aeqd \max_X \f{\nor{\mR_k(X)}_1}{\nor{X}_1}.
}
\end{lemma}

\noindent{\textbf{Remark.}}  From \re{eta_goal_pre}, we obtain
\bea{
\eta(\rho, A, \mE, \mR)^2 \ge \eta(\rho, A, \tl \Lm_{S^\pr}, \mR)^2
-\Big(\nor{\mR_2}_1+2 \nor{\mR_1}_1\nor{A} +2 \nor{A}^2 \Big) T_\RM{LRB}  .
}
By putting 
\bea{
\min_\mR \eta(\rho, A, \mE, \mR) \aeq \eta(\rho, A, \mE, \mR^\ast),
}
we obtain
\bea{
\eta(\rho, A, \mE)^2 \ge \eta(\rho, A, \tl \Lm_{S^\pr}, \mR^\ast)^2
-\Big(\nor{\mR_2^\ast}_1+2 \nor{\mR_1^\ast}_1\nor{A} +  2\nor{A}^2 \Big) T_\RM{LRB} . \la{eta_goal_pre2} 
}

\begin{proof}
We prove \re{eta_goal_pre}. First, we calculate
\bea{
\sig_\theta \defe \Lm_{S^\pr} \circ \calU_{A,\theta}(\rho \otimes \ke{\ep}\bra{\ep})
}
where $\ep=\pm$. 
Because of
\bea{
\calU_{A,\theta}(\rho \otimes \ke{\ep}\bra{\ep}) 
\aeq \rho \otimes \ke{\ep}\bra{\ep} -i\theta(A\rho \otimes \ke{-\ep}\bra{\ep}-\rho A\otimes \ke{\ep}\bra{-\ep} ) \no\\
&\hs{5mm}+\theta^2 \Big[ A\rho A  \otimes \ke{-\ep}\bra{-\ep} -\half (A^2\rho+\rho A^2 ) \otimes \ke{\ep}\bra{\ep} \Big] +O(\theta^3),
}
we obtain
\bea{
\sig_\theta \aeq \sig^{(0)} + \theta \sig^{(1)} + \theta^2 \sig^{(2)} + O(\theta^3) , \\ 
\sig^{(0)} \aeqd \Lm_{S^\pr}(\rho) \otimes \ke{\ep}\bra{\ep} ,\\ 
\sig^{(1)}  \aeqd -i \big[ \Lm_{S^\pr}(A\rho) \otimes \ke{-\ep}\bra{\ep}-\Lm_{S^\pr}(\rho A)\otimes \ke{\ep}\bra{-\ep} \big] ,\\
\sig^{(2)}  \aeqd \Lm_{S^\pr}(A\rho A)  \otimes \ke{-\ep}\bra{-\ep} -\half \Lm_{S^\pr}(A^2\rho+\rho A^2 ) \otimes \ke{\ep}\bra{\ep}
}
and 
\bea{
\dl \Lm_{S^\pr} \circ \calU_{A,\theta}(\rho \otimes \ke{\ep}\bra{\ep})\aeq \dl\sig^{(0)}+\theta \dl\sig^{(1)}+\theta^2 \dl\sig^{(2)}+O(\theta^3)  ,\\
\dl\sig^{(0)} \aeq  \dl \Lm_{S^\pr}(\rho) \otimes \ke{\ep}\bra{\ep} ,\\
\dl\sig^{(1)} \aeq -i\Big[ \dl \Lm_{S^\pr}(A\rho) \otimes \ke{-\ep}\bra{\ep}-\dl \Lm_{S^\pr}(\rho A)\otimes \ke{\ep}\bra{-\ep} \Big] ,\\
 \dl\sig^{(2)} \aeq  \dl \Lm_{S^\pr}(A\rho A)  \otimes \ke{-\ep}\bra{-\ep} 
 -\half \dl \Lm_{S^\pr}(A^2\rho+\rho A^2 ) \otimes \ke{\ep}\bra{\ep}
}
where 
\bea{
\dl \Lm_{S^\pr} \defe \Lm_{S^\pr} - \tl \Lm_{S^\pr}.
}
Next,  we calculate $S_\theta\defe \mR (\sig_\theta)$. 
Expanding $S_\theta$ as 
\bea{
S_\theta \aeq S_0+\theta S_1 + \theta^2 S_2  + O(\theta^3),
}
we obtain
\bea{
S_0 \aeq \mR_0(\sig^{(0)})=\mR_0( \Lm_{S^\pr}(\rho) \otimes \ke{\ep}\bra{\ep} ) ,\\
S_1 \aeq \mR_1(\sig^{(0)})+\mR_0(\sig^{(1)}) ,\\
S_2 \aeq \mR_2(\sig^{(0)})+\mR_1(\sig^{(1)})+\mR_0(\sig^{(2)}).
}
The fidelity $F(\ke{\ep}\bra{\ep}, S_\theta)^2 $ is given by
\bea{
F(\ke{\ep}\bra{\ep}, S_\theta)^2 \aeq \tr_Q[\ke{\ep}\bra{\ep} S_0]
+ \theta\tr_Q[\ke{\ep}\bra{\ep} S_1]+ \theta^2\tr_Q[\ke{\ep}\bra{\ep} S_2]+ O(\theta^3).
}
Here, we suppose that 
\bea{
F(\ke{\ep}\bra{\ep}, S_\theta)^2 \aeq 1-\theta^2 \eta_\ep^2(\rho, A, \mE, \mR)+ O(\theta^3).
}
Then, the disturbance is given by
\bea{
 \eta(\rho, A, \mE, \mR)^2 \aeq  \sum_{\ep=\pm} \half \eta_\ep^2(\rho, A, \mE, \mR) \no\\
 \aeq -\half  \sum_{\ep=\pm}  \tr_Q[\ke{\ep}\bra{\ep} S_2] \no\\
 \aeq -\half  \sum_{\ep=\pm}  \Big[\tr_Q[\ke{\ep}\bra{\ep} \mR_2(\sig^{(0)})] +  \tr_Q[\ke{\ep}\bra{\ep} \mR_1(\sig^{(1)})] \no\\
&\spa +\tr_Q[\ke{\ep}\bra{\ep} \mR_0(\sig^{(2)})]  \Big].
}
Thus, we obtain
\bea{
& \eta(\rho, A, \mE, \mR)^2 -  \eta(\rho, A, \tl \Lm_{S^\pr}, \mR)^2 \no\\
 \aeq -\half  \sum_{\ep=\pm}  \Big[\tr_Q[\ke{\ep}\bra{\ep} \mR_2(\dl \sig^{(0)})] +  \tr_Q[\ke{\ep}\bra{\ep} \mR_1(\dl \sig^{(1)})] 
 +\tr_Q[\ke{\ep}\bra{\ep} \mR_0(\dl \sig^{(2)})]  \Big].
}
This leads to 
\bea{
&\abs{ \eta(\rho, A, \mE, \mR)^2 -  \eta(\rho, A, \tl \Lm_{S^\pr}, \mR)^2} \no\\
\aeqle \half  \sum_{\ep=\pm}  \Big[ \abs{\tr_Q[\ke{\ep}\bra{\ep} \mR_2(\dl \sig^{(0)})]} + \abs{ \tr_Q[\ke{\ep}\bra{\ep} \mR_1(\dl \sig^{(1)})] }
 +\abs{\tr_Q[\ke{\ep}\bra{\ep} \mR_0(\dl \sig^{(2)})] } \Big].
}
Here, 
\bea{
\abs{\tr_Q[\ke{\ep}\bra{\ep} \mR_2(\dl \sig^{(0)})]} \aeqle \nor{\ke{\ep}\bra{\ep} } \cdot \nor{\mR_2(\dl \sig^{(0)})}_1 \no\\
\aeq  \nor{\mR_2(\dl \sig^{(0)})}_1 \no\\
\aeqle \nor{\mR_2}_1 \nor{\dl \sig^{(0)}}_1
}
holds where 
\bea{
\nor{\Psi}_1 \aeqd \max_X \f{\nor{\Psi(X)}_1}{\nor{X}_1} 
}
is the trace norm of the super-operator $\Psi$. 
Using \re{IKS}, we obtain
\bea{
\nor{\dl \sig^{(0)}}_1\aeq \nor{\ke{\ep}\bra{\ep} }_1 \cdot \nor{\dl \Lm_{S^\pr}(\rho)}_1 \no\\
\aeqle T_\RM{LRB}.
}
Thus, 
\bea{
\abs{\tr_Q[\ke{\ep}\bra{\ep} \mR_2(\dl \sig^{(0)})]} \aeqle  \nor{\mR_2}_1 T_\RM{LRB}
}
holds. 
Similarly, we obtain
\bea{
\abs{\tr_Q[\ke{\ep}\bra{\ep} \mR_1(\dl \sig^{(1)})]} \aeqle \nor{\mR_1(\dl \sig^{(1)})}_1 \no\\
\aeqle \nor{\mR_1}_1 \nor{\dl \sig^{(1)}}_1 \no\\
\aeqle \nor{\mR_1}_1\Big(\nor{\dl \Lm_{S^\pr}(A\rho)}_1+\nor{\dl \Lm_{S^\pr}(\rho A)}_1 \Big) \no\\
\aeqle \nor{\mR_1}_1(\nor{A\rho}_1 +\nor{\rho A}_1 )T_\RM{LRB} \no\\
\aeqle 2 \nor{\mR_1}_1\nor{A} T_\RM{LRB} 
}
and 
\bea{
\abs{\tr_Q[\ke{\ep}\bra{\ep} \mR_0(\dl \sig^{(2)})] } \aeqle  \nor{\mR_0(\dl \sig^{(2)})}_1 \no\\
\aeqle  \nor{\mR_0}_1 \nor{\dl \sig^{(2)}}_1 \no\\
\aeqle \nor{\mR_0}_1 \Big( \nor{\dl \Lm_{S^\pr}(A\rho A) }_1 + \half \nor{\dl \Lm_{S^\pr}(A^2\rho+\rho A^2 ) }_1  \Big) \no\\
\aeqle 2 \nor{\mR_0}_1\nor{A}^2 T_\RM{LRB} .
}
Thus, we obtain 
\bea{
\abs{ \eta(\rho, A, \mE, \mR)^2 -  \eta(\rho, A, \tl \Lm_{S^\pr}, \mR)^2} 
\le \Big(\nor{\mR_2}_1+2 \nor{\mR_1}_1\nor{A} + 2 \nor{\mR_0}_1\nor{A}^2 \Big) T_\RM{LRB}.
}
Because $\mR_0$ is a CPTP map and 
\bea{
\nor{\Ga}_1 = 1 \la{nor_1_Ga=1}
}
holds for a positive and trace-preserving map $\Ga$ \cite{Watrous_text}, $\nor{\mR_0}_1=1$ holds. 
Then, we obtain \re{eta_goal_pre}.
\end{proof}

In the same way, we can obtain the following lemma:
\begin{lemma}
The relation
\bea{
\abs{\ep(\rho, A, \mE, \mR)^2-\ep(\rho, A, \tl \Lm_{P}, \mR)^2} 
\aeqle \Big(\nor{\mR_2}_1+2 \nor{\mR_1}_1\nor{A} +2 \nor{A}^2 \Big) T_\RM{LRB}
}
holds. 
$\tl \Lm_{P}$ is $\tl \Lm_{S^\pr}$ for $S^\pr=P$. 
\end{lemma}

\subsection{Lemma for extension of SIQ trade-off relation} \la{s5}

In this subsection, we derive the following Lemma \ref{l_key}, which is an expansion of Lemma 1 of Ref.\cite{Tajima22}.

\begin{lemma} \la{l_key}
Let $Q$ be a projective operator on $K$ ($Q^2=Q$) and $P$ be a non-negative operator on $K^\pr$ satisfying $P^2 \le P$.
We consider a CPTP map $\mL$ from $K$ to $K^\pr$ defined as 
\bea{
\mL(\bu) \aeq \tr_{R^\pr}[V\bu \otimes \rho_R V\dg],
}
where $V$ is an unitary operator on $KR=K^\pr R^\pr$. 
We suppose that
\bea{
\expt{\mL\dg(P)}_{(1-Q)\rho(1-Q)}+\expt{1-\mL\dg(P)}_{Q\rho Q}=\ep^2.
}
$\rho$ is a state on $K$. 
We introduce
\bea{
Z\aeqd V\dg(X_{K^\pr}+X_{R^\pr})V-(X_K+X_R) ,\la{def_Z}\\
W_K\aeqd \tr_R(\rho_R Z).
}
Here, $X_\bu$ $(\bu=S,E,S^\pr,E^\pr)$ is a Hermitian operator on $\bu$. 
Then, the following relation holds:
\bea{
\abs{\expt{ [Q, Y+ W_K] }_\rho} \le \ep (\Dl^\pr+\sqrt{\mF}) \la{L1}
}
where $Y\defe X_K-\mL\dg (X_{K^\pr})$, $\mF\defe \mF_{\rho_R}(X_R)$, and $\Dl^\pr \defe \Dl+\sqrt{\mF_{\rho\otimes \rho_E}(Z)}$ with
\bea{
\Dl \defe \sqrt{\mF_{\rho \otimes \rho_R}(X_K-V\dg X_{K^\pr}V)}.
}
The relation
\bea{
\nor{W_K} \le \nor{Z} \la{W_Z}
}
also holds.

\end{lemma}

\begin{proof}
To prove \re{L1}, we use the following improved version the Robertson uncertainty relation \cite{Frowis_unc, TN}
\bea{
\abs{\expt{[O_1, O_2]}_\sig}^2 \le \mF_\sig(O_1)V_\sig(O_2) .\la{key_R}
}
Here, $V_\sig(O)\defe \tr(\sig O^2)-[\tr(\sig O)]^2$. 
We set
\bea{
O_2 \aeq N \defe V\dg P\otimes 1_{R^\pr}V-Q\otimes 1_R ,\\
O_1 \aeq V\dg 1_{K^\pr}\otimes X_{R^\pr}V ,
}
and $\sig = \rho \otimes \rho_R$. 
We prove
\bea{
V_{\rho \otimes \rho_R}(N) \aeqle \ep^2 ,\la{L1_1}\\
\expt{ [O_1, N]}_{\rho \otimes \rho_R}\aeq \expt{ [Q, Y+ W_K]}_\rho, \la{L1_2}\\
\mF_\sig(O_1) \aeqle (\Dl^\pr+\sqrt{\mF})^2, \la{L1_3}
}
which lead to \re{L1}.

First, we prove \re{L1_3}:
\bea{
\sqrt{\mF_\sig(O_1)} \aeq  \sqrt{\mF_{\rho \otimes \rho_R}(X_K+X_R-V\dg X_{K^\pr}V+Z)} \no\\
\aeqle \sqrt{\mF_{\rho \otimes \rho_R}(X_K-V\dg X_{K^\pr}V)}+\sqrt{\mF_{\rho \otimes \rho_R}(X_R)}
+\sqrt{\mF_{\rho \otimes \rho_R}(Z)} \no\\
\aeq  \sqrt{\mF_{\rho \otimes \rho_R}(X_K-V\dg X_{K^\pr}V)}+\sqrt{\mF_{\rho_R}(X_R)}
+\sqrt{\mF_{\rho \otimes \rho_R}(Z)} \no\\
\aeq \Dl + \sqrt{\mF_{\rho_R}(X_R)}+\sqrt{\mF_{\rho \otimes \rho_R}(Z)}.
}
Here, we used \re{def_Z} in the first line, 
\bea{
\sqrt{\mF_\sig(W+W^\pr)} \le \sqrt{\mF_\sig(W)}+\sqrt{\mF_\sig(W^\pr)} \la{L1_3b}
}
in the second, and $\mF_{\rho \otimes \rho_R}(1_K \otimes X_R)=\mF_{\rho_R}(X_R)$ in the third line.
\re{L1_3b} is derived as follow:
\bea{
\mF_\sig(W+W^\pr) \aeq  \expt{L_W+L_{W^\pr}, L_W+L_{W^\pr} }_\sig^\RM{SLD} \no\\
\aeq \expt{L_W, L_W}_\sig^\RM{SLD}+2\expt{L_W, L_{W^\pr}}_\sig^\RM{SLD}+\expt{L_{W^\pr}, L_{W^\pr} }_\sig^\RM{SLD} \no\\
\aeqle \expt{L_W, L_W}_\sig^\RM{SLD}+2 \sqrt{\expt{L_W, L_W}_\sig^\RM{SLD}\expt{L_{W^\pr}, L_{W^\pr}}_\sig^\RM{SLD}} 
+\expt{L_{W^\pr}, L_{W^\pr} }_\sig^\RM{SLD} \no\\
\aeq \big(\sqrt{\expt{L_W, L_W}_\sig^\RM{SLD} }+ \sqrt{\expt{L_{W^\pr}, L_{W^\pr}}_\sig^\RM{SLD} }\big)^2 \no\\
\aeq \big(\sqrt{\mF_\sig(W)}+\sqrt{\mF_\sig(W^\pr)}\big)^2.
}

Next, we prove \re{L1_1}:
\bea{
V_{\rho \otimes \rho_R}(N) \aeqle \tr[\rho \otimes \rho_R N^2] \no\\
\aeq \tr[\rho \otimes \rho_R (V\dg P^2V-V\dg P VQ-QV\dg P V+Q)] \no\\
\aeqle \tr[\rho \otimes \rho_R (V\dg PV-V\dg P VQ-QV\dg P V+Q)] \no\\
\aeq  \tr[\rho \otimes \rho_R \{(1-Q) V\dg PV(1-Q)
+Q V\dg (1-P )V\dg Q  \}]  \no\\
\aeq \expt{\mL\dg(P)}_{(1-Q)\rho_S(1-Q)}+\expt{\mL\dg(1-P)}_{Q\rho_SQ} \no\\
\aeq \ep^2.
}
We used $Q^2=Q$ in the second line and $P^2 \le P$ in the third line.
In the fifth line, we used 
\bea{
\expt{\mL\dg(A)}_{B\rho B} \aeq \tr_K[\mL\dg(A)B\rho B] \no\\
\aeq \tr_{K^\pr}[A\mL(B\rho B)]  \no\\ 
\aeq \tr_{K^\pr}[A\tr_{R^\pr}\{V(B\rho B \otimes \rho_R)V\dg\}] \no\\
\aeq \tr[AV(B\rho B \otimes \rho_R)V\dg] \no\\
\aeq \tr[\rho\otimes \rho_R (B V\dg AV  B)].
}

Next, we prove \re{L1_2}. 
First, 
\bea{
 [O_1, N]\aeq  [V\dg 1_{K^\pr}\otimes X_{R^\pr}V, V\dg P\otimes 1_{R^\pr}V-Q\otimes 1_R] \no\\
\aeq  [V\dg 1_{K^\pr}\otimes X_{R^\pr}V, -Q\otimes 1_R] \no\\
\aeq [X_K+X_R-V\dg X_{K^\pr}V+Z, -Q\otimes 1_R] \no\\
\aeq [X_K-V\dg X_{K^\pr}V+Z, -Q\otimes 1_R]
}
holds. Then, we obtain
\bea{
\expt{ [O_1, N]}_{\rho \otimes \rho_R}\aeq \expt{ [X_K-V\dg X_{K^\pr}V+Z, -Q\otimes 1_R]}_{\rho \otimes \rho_R} \no\\
\aeq \expt{ [Q, Y+W_K]}_\rho
}
using
\bea{
\expt{ AYB }_\rho \aeq \tr_K(\rho A[X_K-\mL \dg (X_{K^\pr})]B) \no\\
\aeq \tr_K(\rho AX_KB)-\tr_{K^\pr}(X_{K^\pr}\mL(B \rho A)) \no\\
\aeq \tr_K(\rho AX_KB)-\tr(X_{K^\pr}V(B \rho A) \otimes \rho_R V\dg ) \no\\
\aeq \tr[\rho \otimes \rho_R(AX_KB-AV\dg X_{K^\pr}V B )]
}
and $\expt{ [Q, Z]}_{\rho \otimes \rho_R} =\expt{ [Q, W_K]}_{\rho}$. 

To prove \re{W_Z}, we use the formula
\bea{
\nor{W_K} \aeq \max_{\nor{A_K}_1 \le 1}\abs{\tr_K(A_KW_K)}.
}
This equation lends to
\bea{
\nor{W_K} \aeq \max_{\nor{A_K}_1 \le 1}\abs{\tr_{KR}(A_K\otimes\rho_RZ)}.
}
Because of
\bea{
\abs{\tr(AB)}\le \nor{B} \cdot \nor{A}_1,
}
we obtain
\bea{
\nor{W_K} \le \max_{\nor{A_K}_1 \le 1}\nor{Z}  \cdot \nor{\rho_R}_1\cdot \nor{A_K}_1 = \nor{Z} .
}
\end{proof}

\subsection{An extension of SIQ trade-off relation} \la{s6}

In this subsection, we prove the following Theorem \ref{theo_1}, which is an expansion of Theorem \ref{theo_0}. 

\begin{theorem} \la{theo_1}

We consider a CPTP map $\mL$ from $K$ to $K^\pr$ given by $\mL(\bu)=\tr_{R^\pr}(U \bu \otimes \rho_R U\dg)$
and suppose that $Z=U\dg(X_{K^\pr}+X_{R^\pr})U-(X_K+X_R)$ holds. 
Here, $X_\bu$ $(\bu=K,R,K^\pr,R^\pr)$ is a Hermitian operator on the system $\bu$ and $U$ is an unitary operator on $KR=K^\pr R^\pr$.  
Then, for a test ensemble $(\{(p_k,\rho_k)\}_{k=1}^n)$ that $\{ \rho_k \}$ are orthogonal to each other, the relation
\bea{
\dl(\mL, \Om) \sqrt{1-\min_k p_k} \ge \f{\mC_W}{\Dl^\pr+\sqrt{\mF}} \la{SIQ_theorem}
}
holds. Here, $\dl(\mL, \Om) \defe \min_\mR \dl(\mL, \mR, \Om)$.
$\mF \defe \mF_{\rho_R}(X_R)$ and $\Dl^\pr\defe \Dl+\Dl_Z$ with  
\bea{
\Dl \aeqd \max_{\rho \in \cup_k \RM{supp}(\rho_k)}\sqrt{\mF_{\rho\otimes \rho_R}(X_K-U\dg X_{K^\pr}U)} ,\\
\Dl_Z \aeqd \max_{\rho \in \cup_k \RM{supp}(\rho_k)}\sqrt{\mF_{\rho\otimes \rho_R}(Z)}
}
where the maximum runs over the subspace spanned by the supports of the test states $\{ \rho_k \}$.
$\mC_W$ is defined by
\bea{
\mC_W \aeqd \sqrt{\sum_{k \ne k^\pr}p_k p_{k^\pr} \tr_K[\rho_k (Y_K+W_K) \rho_{k^\pr} (Y_K+W_K)] } 
}
with $Y_K \defe X_K-\mL\dg (X_{K^\pr})$ and $W_K \defe \tr_R(\rho_RZ)$.

\end{theorem}

\begin{proof}
We prove \re{SIQ_theorem}. 
We introduce projective operators $\{Q_k\}_{k=1}^n$ that distinguish $\{\rho_k\}_{k=1}^n$, i.e., $Q_k\rho_{k^\pr}=\dl_{kk^\pr}\rho_{k^\pr}=\rho_{k^\pr}Q_k$. 
$Q_kQ_{k^\pr}=Q_k \dl_{kk^\pr}$ holds. 
We also introduce $\{\tl Q_\al\}_{\al=1}^d$ $(d\ge n)$ such that 
$\tl Q_\al \tl Q_{\al^\pr}=\tl Q_\al \dl_{\al\al^\pr}$, 
$\sum_\al \tl Q_\al =1_K$, and $\tl Q_k =Q_k $ ($k=1,2, \cdots, n$).
We introduce a CPTP map $\mQ$ as $\mQ(\bu)\defe \sum_\al \tr_K[\tl Q_\al \bu] \ke{\al}\bra{\al}_A$
where  $\{\ke{\al}_A\}$ is an orthonormal basis of system $A$. 
The monotonicity of $D_F$ leads to
\bea{
\dl_k \aeqge D_F(\mQ(\rho_k), \mQ\circ\mR \circ\mL(\rho_k) ) \no\\
\aeq D_F(\ke{k}\bra{k}_A, \sig_k )=\sqrt{1-s_k} \la{mon_Q}
}
where $\sig_k \defe \mQ\circ\mR \circ\mL(\rho_k)$ and $s_k \defe \bra{k}\sig_k \ke{k}$. 
\re{mon_Q} implies 
\bea{
s_k \aeq \tr_K[Q_k \mR \circ\mL(\rho_k)] \no\\
\aeq \tr_K[\rho_k \mL \dg \circ \mR\dg (Q_k) ]\ge 1-\dl_k^2. \la{mon_Q1}
}
We introduce
\bea{
s_{k,k^\pr} \defe \tr_K[\rho_k \mL \dg \circ \mR\dg (Q_{k^\pr}) ].
}
Because of $\sum_\al \tl Q_\al=1_K$ and $\mL \dg \circ \mR\dg (1_K)=1_K$, 
$\sum_{k^\pr}s_{k,k^\pr}\le 1$ holds. 
This relation and \re{mon_Q1} lead to
\bea{
s_{k,k^\pr}  \le 1-s_{k,k} =1-s_k \le \dl_k^2 \ \ (k^\pr \ne k). \la{mon_Q2}
}
Now, let us a spectral decomposition of $\rho_k$ as
\bea{
\rho_k=\sum_l q_l^{(k)}\psi_l^{(k)}
}
with $\psi\defe \ke{\psi}\bra{\psi}$, and define
\bea{
1-\dl_{(k),l}^2 \aeqd \tr_K(\psi_l^{(k)}\mL \dg \circ \mR\dg (Q_k) ),\\
\dl_{(k),[k^\pr],l}^2 \aeqd \tr_K(\psi_l^{(k)}  \mL \dg \circ \mR\dg (Q_{k^\pr}) ).
}
Then, \re{mon_Q1} and \re{mon_Q2} become
\bea{
\sum_l q_l^{(k)}\dl_{(k),l}^2 \aeqle \dl_k^2 ,\\
\sum_l q_l^{(k)}\dl_{(k),[k^\pr],l}^2 \aeqle \dl_k^2  \ \ (k^\pr \ne k).
}

We introduce 
\bea{
\ke{\psi_{k,k^\pr,l,l^\pr}(\theta)} \defe \f{\ke{\psi_l^{(k)}}+e^{i\theta}\ke{\psi_{l^\pr}^{(k^\pr)}}}{\sqrt{2}}.
}
For $k^\pr \ne k$, 
\bea{
Q_k\psi_{k,k^\pr,l,l^\pr}(\theta) Q_k \aeq \half \psi_l^{(k)} ,\\
(1-Q_k)\psi_{k,k^\pr,l,l^\pr}(\theta) (1-Q_k) \aeq \half \psi_{l^\pr}^{(k^\pr)}
}
hold. Here, we used $Q_k\ke{\psi_{l}^{(k^\pr)}} =\dl_{kk^\pr}\ke{\psi_{l}^{(k^\pr)}}$.
By introducing $P_k \defe \mR\dg (Q_k)$, we obtain
\bea{
&\tr_K[(1-\mL\dg(P_k))Q_k\psi_{k,k^\pr,l,l^\pr}(\theta) Q_k ]+
\tr_K[\mL\dg(P_k)(1-Q_k)\psi_{k,k^\pr,l,l^\pr}(\theta) (1-Q_k)] \no\\
\aeq \half[ \tr_K\{(1-\mL\dg(P_k)) \psi_l^{(k)}\} + \tr_K\{\mL\dg(P_k) \psi_{l^\pr}^{(k^\pr)}  \} \no\\
\aeq \half (\dl_{(k),l}^2+\dl_{(k^\pr),[k],l^\pr}^2).
}
Applying the Lemma \ref{l_key} to the above equation, we obtain
\bea{
\half (\dl_{(k),l}^2+\dl_{(k^\pr),[k],l^\pr}^2) 
\aeqge \f{\abs{\expt{ [Q_k, Y_K+W_K]}_{\psi_{k,k^\pr,l,l^\pr}(\theta)}}^2}{\big(\Dl^\pr_{k,k^\pr,l,l^\pr}+\sqrt{\mF_{\rho_R}(X_R)}\big)^2 }
}
where 
\bea{
\Dl^\pr_{k,k^\pr,l,l^\pr} \aeqd \sqrt{\mF_{\psi_{k,k^\pr,l,l^\pr}(\theta) \otimes \rho_R}(X_K-U\dg X_{K^\pr}U)}
+\sqrt{\mF_{\psi_{k,k^\pr,l,l^\pr}(\theta) \otimes \rho_R}(Z)} \no\\
\aeqle \Dl+\Dl_Z.
}
Because
\bea{
\expt{ [Q_k, Y_K+W_K]}_{\psi_{k,k^\pr,l,l^\pr}(\theta)} \aeq \half e^{i\theta} \bra{\psi_l^{(k)}}(Y_K+W_K)\ke{\psi_{l^\pr}^{(k^\pr)}}-(\mbox{c.c.})
}
holds, by setting $e^{i\theta_\ast}=ie^{-i\eta}$ with
\bea{
\bra{\psi_l^{(k)}}(Y_K+W_K)\ke{\psi_{l^\pr}^{(k^\pr)}} \aeq e^{i\eta}\abs{\bra{\psi_l^{(k)}}(Y_K+W_K)\ke{\psi_{l^\pr}^{(k^\pr)}}},
}
we obtain
\bea{
\expt{ [Q_k, Y_K+W_K]}_{\psi_{k,k^\pr,l,l^\pr}(\theta_\ast)} \aeq i \abs{ \bra{\psi_l^{(k)}}(Y_K+W_K)\ke{\psi_{l^\pr}^{(k^\pr)}}}.
}
Then, we obtain \re{SIQ_theorem} as follows:
\bea{
\f{\mC_W^2}{(\Dl^\pr+\sqrt{\mF_{\rho_R}(X_R)})^2} 
\aeq \sum_{k\ne k^\pr}\sum_{l,l^\pr} p_kp_{k^\pr} q_l^{(k)}q_{l^\pr}^{(k^\pr)} 
\f{\abs{ \bra{\psi_l^{(k)}}(Y_K+W_K)\ke{\psi_{l^\pr}^{(k^\pr)}}}^2}{(\Dl^\pr+\sqrt{\mF_{\rho_R}(X_R)})^2} \no\\
\aeqle \sum_{k\ne k^\pr}\sum_{l,l^\pr} p_kp_{k^\pr} q_l^{(k)}q_{l^\pr}^{(k^\pr)} \f{\dl_{(k),l}^2+\dl_{(k^\pr),[k],l^\pr}^2}{2} \no\\
\aeqle \sum_{k\ne k^\pr} p_kp_{k^\pr} \dl_k^2 \no\\
\aeq \sum_{k} p_k(1-p_k) \dl_k^2 \no\\
\aeq \dl^2-\sum_k p_k^2 \dl_k^2 \no\\
\aeqle (1-\min_k p_k)\dl^2.
}
\end{proof}

\subsection{Lower bounds of disturbance and error}  \la{s7}

To evaluate $\dl(\mL_{\psi,A,\theta,\mN}, \mR, \Om_{0})$, the following relation is useful. 

\begin{lemma} \la{l_Emori}

The following relation holds \cite{ET2023}: 
\bea{
\dl(\mL_{\psi,A,\theta,\mN}, \mR, \Om_{0}) \ge \dl(\mN \otimes 1_Q, \Om_{1/2,\Psi_{\pm,\theta}}). \la{benri}
}
Here, $\Om_{1/2,\Psi_{\pm,\theta}}$ denotes the test ensemble $\{(1/2, \Psi_{+,\theta}), (1/2, \Psi_{-,\theta})\} $
where $\ke{\Psi_{\ep,\theta}}\defe e^{-i\theta A\otimes \sig_z}\ke{\psi} \otimes \ke{\ep}$.

\end{lemma} 

\begin{proof}
We prove \re{benri}.
We take $\mR^\ast$ such that 
\bea{
\dl(\mL_{\psi,A,\theta,\mN}, \Om_{0})=\dl(\mL_{\psi,A,\theta,\mN}, \mR^\ast, \Om_{0}).
}
Then, 
\bea{
\sum_{\ep=\pm} \half \bra{\ep}\mR^\ast \circ \mL_{\psi,A,\theta,\mN}(\ke{\ep}\bra{\ep}_Q)\ke{\ep}
=1-\dl(\mL_{\psi,A,\theta,\mN}, \Om_{0})^2
}
holds. 
Because of $\mL_{\psi,A,\theta,\mN}(\ke{\ep}\bra{\ep}_Q)=(\mN \otimes 1_Q)(\Psi_{\ep,\theta})$,
we obtain
\bea{
\sum_{\ep=\pm} \half \bra{\ep}\mR^\ast \circ (\mN \otimes 1_Q)(\Psi_{\ep,\theta})\ke{\ep}
=1-\dl(\mL_{\psi,A,\theta,\mN}, \Om_{0})^2.
}
By introducing $\mT_\Psi(\bu) \defe \sum_{\ep=\pm}\bra{\ep}\bu \ke{\ep} \Psi_{\ep,\theta}$, we obtain
\bea{
1-\dl(\mL_{\psi,A,\theta,\mN}, \Om_{0})^2 
\aeq \half \bra{\Psi_{\ep,\theta}}\mT_\Psi \circ \mR^\ast \circ (\mN \otimes 1_Q) (\Psi_{\ep,\theta})\ke{\Psi_{\ep,\theta}} \no\\
\aeq 1-\dl(\mN \otimes 1_Q, \mT_\Psi \circ \mR^\ast, \Om_{1/2,\Psi_{\pm,\theta}})^2.
}
Thus, 
\bea{
\dl(\mL_{\psi,A,\theta,\mN}, \Om_{0}) \aeq \dl(\mN \otimes 1_Q, \mT_\Psi \circ \mR^\ast, \Om_{1/2,\Psi_{\pm,\theta}})  \no\\
\aeqge \dl(\mN \otimes 1_Q, \Om_{1/2,\Psi_{\pm,\theta}}) 
}
holds. 
\end{proof}

Lemma \ref{l_Emori} and Theorem \ref{theo_1} lead to the following lemma. 

\begin{lemma} \la{l_pure}

We consider a CPTP map $\mN$ from $\al$ to $\al^\pr$ given by $\mN(\bu)=\tr_{\be^\pr}(U \bu \otimes \rho_\be U\dg)$
and suppose that $Z=U\dg(X_{\al^\pr}+X_{\be^\pr})U-(X_\al+X_\be)$ holds. 
Here, $X_\bu$ $(\bu=\al,\be,\al^\pr,\be^\pr)$ is a Hermitian operator on the system $\bu$ and $U$ is an unitary operator on $\al \be=\al^\pr \be^\pr$. 
Then, the relation 
\bea{
\eta(\psi, A, \mN, \mR) \ge \f{\abs{\expt{[Y_\al+W_\al,A]}_{\psi}}}{\sqrt{\mF}+\Dl_F(\psi)+\Dl_1} \la{key_pre}
}
holds. 
Here, $Y_\al\defe X_\al-\mN\dg(X_{\al^\pr})$ and $W_\al\defe \tr_\be(\rho_\be Z)$. 
$\mF \defe \mF_{\rho_\be}(X_\be) $. 
$\Dl_F(\rho)$ and $\Dl_1$  are defined by
\bea{
\Dl_F(\rho) \aeqd \sqrt{\mF_{\rho}(X_\al)}+2\sqrt{V_{\mN(\rho)}(X_{\al^\pr})} ,\\
\Dl_1\aeqd \max_{\rho_\al}\sqrt{\mF_{\rho_\al \otimes \rho_\be}(Z)} .
}

\end{lemma} 

\begin{proof}
We prove \re{key_pre}.
\re{SIQ_theorem} and \re{benri} lead to 
\bea{
\dl(\mL_{\psi,A,\theta,\mN}, \mR, \Om_{0})\ge \dl(\mN \otimes 1_Q, \Om_{1/2,\Psi_{\pm,\theta}})  \ge \f{\sqrt{2}\mC_W}{\sqrt{\mF}+\Dl^\pr}.
}
Thus, we obtain
\bea{
\eta(\psi,A,\mN,\mR) \ge \lim_{\theta \to +0} \f{\sqrt{2}\mC_W}{\theta(\sqrt{\mF}+\Dl^\pr)}
}
where 
\bea{
\mC_W^2 \aeq \sum_{\ep\ne \ep^\pr}\f{1}{4}\abs{\bra{\Psi_{\ep,\theta}}(Y_\al+W_\al)\otimes 1_Q\ke{\Psi_{\ep^\pr,\theta}}}^2 \no\\
\aeq \half \abs{\bra{\Psi_{+,\theta}}(Y_\al+W_\al)\otimes 1_Q\ke{\Psi_{-,\theta}}}^2 \no\\
\aeq \f{\theta^2}{2} \abs{\expt{[Y_\al+W_\al,A]}_{\psi}}^2+O(\theta^3) 
}
and 
\bea{
\Dl^\pr \aeq \max_{\rho \in \cup_\ep \RM{supp}(\ke{\Psi_{\ep,\theta}})}\Big\{\sqrt{\mF_{\rho\otimes \rho_\be}(X_\al-U\dg X_{\al^\pr}U)} +
\sqrt{\mF_{\rho\otimes \rho_\be}(Z)} \Big\} \no\\
\aeqle \max_{\rho \in \cup_\ep \RM{supp}(\ke{\Psi_{\ep,\theta}})}\Big\{\sqrt{\mF_{\rho\otimes \rho_\be}(X_\al)}
 +\sqrt{\mF_{\rho\otimes \rho_\be}(U\dg X_{\al^\pr}U)} +\sqrt{\mF_{\rho\otimes \rho_\be}(Z)} \Big\} \no\\
 \aeq \max_{\rho \in \cup_\ep \RM{supp}(\ke{\Psi_{\ep,\theta}})}\Big\{\sqrt{\mF_{\rho\otimes \rho_\be}(X_\al)}
 +\sqrt{\mF_{U\rho\otimes \rho_\be U \dg}( X_{\al^\pr})} +\sqrt{\mF_{\rho\otimes \rho_\be}(Z)} \Big\} \no\\
 \aeq \max_{\rho \in \cup_\ep \RM{supp}(\ke{\psi} \otimes \ke{\ep})}\Big\{\sqrt{\mF_{\rho\otimes \rho_\be}(X_\al)}
 +\sqrt{\mF_{U\rho\otimes \rho_\be U \dg}( X_{\al^\pr})} +\sqrt{\mF_{\rho\otimes \rho_\be}(Z)} \Big\}+o(1) \no\\
 \aeq  \sqrt{\mF_{\psi}(X_\al)}+\sqrt{\mF_{U\psi \otimes \rho_\be U\dg}(X_{\al^\pr})}+\sqrt{\mF_{\psi \otimes \rho_\be}(Z)}+o(1) \no\\
 \aeqle  \sqrt{\mF_{\psi}(X_\al)}+2\sqrt{V_{\mN(\psi)}(X_{\al^\pr})}+\sqrt{\mF_{\psi \otimes \rho_\be}(Z)}+o(1). \la{dl^pr}
}
We used \re{L1_3b} in the second line of \re{dl^pr}, 
and 
\bea{
\mF_{U\psi \otimes \rho_\be U\dg}(X_{\al^\pr}) \aeqle 4V_{U\psi \otimes \rho_\be U\dg}(X_{\al^\pr}) \no\\
\aeq  4V_{\tr_{\be^\pr}(U\psi \otimes \rho_\be U\dg)}(X_{\al^\pr}) \no\\
\aeq 4V_{\mN(\psi)}(X_{\al^\pr})
}
in the last line of \re{dl^pr}.
Thus, we obtain \re{key_pre}.
\end{proof}

Lemma \ref{l_pure} leads to the following lemma which is an extension of Theorem 1 of Ref.\cite{ET2023}. 

\begin{lemma} \la{l_mixed}
The following relation holds: 
\bea{
\eta(\rho, A, \mN, \mR) \ge \f{\abs{\expt{[Y_\al+W_\al,A]}_{\rho}}}{\sqrt{\mF}+\Dl_F(\rho)+\Dl_1} .\la{key_pr}
}
The setting and symbols are the same with Lemma \ref{l_pure}. 
\end{lemma}

\begin{proof}
To prove \re{key_pr}, we use a formula
\bea{
\eta(\sum_i r_i \rho_i, A, \mN, \mR)^2 \aeq \sum_i r_i  \eta(\rho_i, A, \mN, \mR)^2 . \la{eta_pro}
}
We can make a pure decomposition $\rho=\sum_j q_j \psi_j$ such that \cite{min_V_Yu}
\bea{
\mF_\rho(X_\al) \aeq 4\sum_j q_j V_{\psi_j}(X_\al). \la{e_57}
}
Using \re{key_pre}, \re{eta_pro}, and \re{e_57}, we obtain
\bea{
\abs{\expt{[Y_\al+W_\al,A]}_\rho} \aeqle \sum_j q_j \expt{[Y_\al+W_\al,A]}_{\psi_j} \no\\
\aeqle \sum_j q_j  \eta(\psi_j,A,\mN, \mR)\Big(\sqrt{\mF}+\Dl_F(\psi_j)+\Dl_1 \Big) \no\\
\aeqle \sqrt{ \sum_j q_j  \eta(\psi_j,A,\mN, \mR)^2 } \Big(
\sqrt{\mF}+\sqrt{ \sum_j q_j \mF_{\psi_j}(X_\al) } \no\\
&\spa  +  2\sqrt{ \sum_j q_j V_{\mN(\psi_j)}(X_{\al^\pr}) }
+ \Dl_1 \Big) \no\\
\aeq \eta(\rho, A,\mN, \mR) \Big(\sqrt{\mF}+\sqrt{ 4\sum_j q_j V_{\psi_j}(X_\al) } 
+ 2\sqrt{ \sum_j q_j V_{\mN(\psi_j)}(X_{\al^\pr}) } + \Dl_1\Big) \no\\
\aeq \eta(\rho, A,\mN, \mR) \Big(\sqrt{\mF}+\sqrt{ \mF_\rho(X_\al) } 
+2 \sqrt{ \sum_j q_j V_{\mN(\psi_j)}(X_{\al^\pr}) } + \Dl_1\Big) \no\\
\aeqle \eta(\rho, A,\mN, \mR) \Big(\sqrt{\mF}+\sqrt{ \mF_\rho(X_\al) } 
+ 2\sqrt{V_{\mN(\rho)}(X_{\al^\pr}) } + \Dl_1\Big) .
}
This leads to \re{key_pr}. 
\end{proof}

Lemma \ref{l_mixed} and Lemma \ref{l_approx} lead to the following theorem.

\begin{theorem} \la{t_main}
The following relation holds:
\bea{
\eta(\rho, A, \mE)^2 
 \aeqge 
 \Big(\f{\mathfrak{R}(\abs{\expt{[Y_S^\pr,A]}_\rho}-\sqrt{\mF_\rho(A)} \{T_\RM{LRB}^\pr+\half \Dl_{H_{S^\pr}}  T_\RM{LRB}\})}
 {\sqrt{\mF^\RM{cost}}+\Dl_F+2T_\RM{LRB}^\pr}\Big)^2 \no\\
&\spa-\Big(\nor{\mR_2^\ast}_1+2 \nor{\mR_1^\ast}_1\nor{A} +  2\nor{A}^2 \Big) T_\RM{LRB} .\la{eta_goal_1}
}
Here, $\mathfrak{R}(x)\defe \max\{0,x\}$, $Y_S^\pr \defe H_S-\Lm_{S^\pr}\dg (H_{S^\pr})$, and 
\bea{
T_\RM{LRB}^\pr \aeqd \f{C}{v} \nor{H_\RM{int}^\pr}  \cdot \nor{H_{\p E_1^\pr}} \cdot
 \abs{I^\pr}\abs{\p E_1^\pr}e^{-\mu d(I^\pr, \p E_1^\pr)} (e^{vt}-vt-1)
}
is an upper bound of $Z_{SE_1}$ defined by \re{con_low_0}.
$I^\pr$ is the supports of $H_\RM{int}^\pr$. 
$\Dl_F \defe \sqrt{\mF_\rho(H_S)}+\Dl_{H_{S^\pr}}$ 
where $\Dl_{H_{S^\pr}}$ is the difference between the largest eigenvalue and the smallest eigenvalue of $H_{S^\pr}$.
$\mF^\RM{cost}$ is defined by
\bea{
\mF^\RM{cost}\aeqd \min_{\rho_{E_1}, U_T} \{\mF_{\rho_{E_1}}(H_{E_1})\Big\vert (\rho_{E_1}, H_{E_1}, H_{E_1^\pr}, U_T)\to \tl \Lm_{S^\pr}  \} \la{def_F_cost}. 
}
Here, $ (\rho_{E_1}, H_{E_1}, H_{E_1^\pr}, U_T)\to \tl \Lm_{S^\pr}$ means 
\bea{
\tl \Lm_{S^\pr}(\bu) \aeq \tr_{E_1^\pr}[U_T\bu \otimes \rho_{E_1} U_T \dg] , \la{min_con_1} \\
Z_{SE_1} \aeq  U_T\dg (H_{S^\pr}+H_{E_1^\pr})U_T-(H_S+H_{E_1})\la{min_con_2} 
}
hold.

\end{theorem}

\noindent{\textbf{Remark.}}  $\sqrt{\mF^\RM{cost}}$ is upper bounded as \re{F_jyougen}.

\begin{proof}
We prove \re{eta_goal_1}. 
For  $ (\rho_{E_1}, H_{E_1}, H_{E_1^\pr}, U_T) $ satisfying \re{min_con_1} and \re{min_con_2}, 
using \re{key_R} and \re{W_Z}, we obtain
\bea{
\abs{\expt{[\tl Y_S^\pr+W_S,A]}_\rho} \aeqge \abs{\expt{[\tl Y_S^\pr,A]}_\rho}-\abs{\expt{[W_S,A]}_\rho} \no\\
\aeqge \abs{\expt{[\tl Y_S^\pr,A]}_\rho}-\sqrt{\mF_\rho(A)V_\rho(W_S)} \no\\
\aeqge \abs{\expt{[\tl Y_S^\pr,A]}_\rho}-\sqrt{\mF_\rho(A)}\nor{W_S} \no\\
\aeqge \abs{\expt{[\tl Y_S^\pr,A]}_\rho}-\sqrt{\mF_\rho(A)}\nor{Z_{SE_1}} . \la{ineq_1}
}
Here, $\tl Y_S^\pr \defe H_S-\tl \Lm_{S^\pr}\dg (H_{S^\pr})$ and $W_S \defe \tr_{E_1}(\rho_{E_1}Z_{SE_1})$. 
$\abs{\expt{[\tl Y_S^\pr,A]}_\rho}$ is calculated as
\bea{
\abs{\expt{[\tl Y_S^\pr,A]}_\rho} \aeqge \abs{\expt{[Y_S^\pr,A]}_\rho}-\abs{\expt{[\dl \Lm_{S^\pr}\dg(H_{S^\pr}),A]}_\rho} \no\\
\aeq \abs{\expt{[Y_S^\pr,A]}_\rho}-\abs{\tr_S(\dl \Lm_{S^\pr}\dg(H_{S^\pr})[A, \rho])} \no\\
\aeq \abs{\expt{[Y_S^\pr,A]}_\rho}-\abs{\tr_{S^\pr}[H_{S^\pr} \dl \Lm_{S^\pr}([A, \rho])]} \no\\
\aeq \abs{\expt{[Y_S^\pr,A]}_\rho}-\abs{\tr_{S^\pr}[(H_{S^\pr}-x1_{S^\pr}) \dl \Lm_{S^\pr}([A, \rho])]} \no\\
\aeqge \abs{\expt{[Y_S^\pr,A]}_\rho}-\nor{H_{S^\pr}-x1_{S^\pr}} \cdot \nor{\dl \Lm_{S^\pr}([A, \rho])}_1 \no\\
\aeq \abs{\expt{[Y_S^\pr,A]}_\rho}-\half \Dl_{H_{S^\pr}} \nor{\dl \Lm_{S^\pr}([A, \rho])}_1 \no\\
\aeqge \abs{\expt{[Y_S^\pr,A]}_\rho}-\half \Dl_{H_{S^\pr}}  \nor{[A, \rho]}_1 T_\RM{LRB} \no\\
\aeqge \abs{\expt{[Y_S^\pr,A]}_\rho}-\half \Dl_{H_{S^\pr}} \sqrt{\mF_\rho(A)} T_\RM{LRB} \la{ineq_2}
}
where $\dl \Lm_{S^\pr} \defe \Lm_{S^\pr}-\tl \Lm_{S^\pr}$. 
$x$ is the average of the largest and smallest eigenvalues of $H_{S^\pr}$. 
Here, we used \re{IKS} and
\bea{
\nor{[A, \rho]}_1 \le \sqrt{\mF_\rho(A)}.
}
To prove this relation, we make a pure decomposition $\rho=\sum_j q_j \psi_j$ such that \cite{min_V_Yu}
\bea{
\mF_\rho(A) \aeq 4\sum_j q_j V_{\psi_j}(A). \la{e_57_2}
}
Using this decomposition, we obtain
\bea{
\nor{[A, \rho]}_1 \aeqle \sum_j q_j \nor{[A, \psi_j]}_1 \no\\
\aeq 2\sum_j q_j \sqrt{V_{\psi_j}(A)} \no\\
\aeqle \sqrt{4\sum_j q_j V_{\psi_j}(A)} \sqrt{\sum_j q_j} \no\\
\aeq  \sqrt{\mF_\rho(A) }.
}
Here, we used 
\bea{
\nor{[A, \psi_j]}_1
\aeq \sqrt{V_{\psi_j}(A)} \tr_S\sqrt{\ke{\psi_j}\bra{\psi_j}+\ke{\psi_j^\perp}\bra{\psi_j^\perp}} \no\\
\aeq 2\sqrt{V_{\psi_j}(A)}
}
where
\bea{
\ke{\psi_j^\perp} \aeqd \f{A-\expt{A}_{\psi_j} }{\sqrt{V_{\psi_j}(A)} }\ke{\psi_j}.
}
\re{ineq_1} and \re{ineq_2} lead to 
\bea{
\eta(\rho, A, \tl \Lm_{S^\pr}, \mR^\ast) \ge \f{\abs{\expt{[Y_S^\pr,A]}_\rho}-\sqrt{\mF_\rho(A)}\{\nor{Z_{SE_1}}+\half \Dl_{H_{S^\pr}} T_\RM{LRB}\}}
{\sqrt{\mF_{\rho_{E_1}}(H_{E_1})}+\tl \Dl_F+\Dl_1} . \la{S160}
}
Here, 
\bea{
\tl \Dl_F \aeqd \sqrt{\mF_\rho(H_S)}+2\sqrt{V_{\tl \Lm_{S^\pr}(\rho)}(H_{S^\pr})} \le \Dl_F ,\\
\Dl_1 \aeqd \max_{\rho_S}\sqrt{\mF_{\rho_S \otimes \rho_{E_1}}(Z_{SE_1})} \le 2\nor{Z_{SE_1}}.  
}
Here, we used
\bea{
\sqrt{V_\sig(X)} = \sqrt{V_\sig(X-x1)} \le \nor{X-x1} =\half \Dl_X \la{vari_X}
}
where $x$ is the average of the largest and smallest eigenvalues of $X$. 
Thus, we obtain 
\bea{
\eta(\rho, A, \tl \Lm_{S^\pr}, \mR^\ast) \ge \f{\abs{\expt{[Y_S^\pr,A]}_\rho}-\sqrt{\mF_\rho(A)}\{\nor{Z_{SE_1}}+\half \Dl_{H_{S^\pr}} T_\RM{LRB}\}}
{\sqrt{\mF^\RM{cost}}+\Dl_F+2\nor{Z_{SE_1}}}. \la{eta_8}
}
We prove
\bea{
\nor{Z_{SE_1}} \aeqle T_\RM{LRB}^\pr \la{Z_upper}
}
in the \res{s_Z}. 
Using \re{eta_8}, \re{Z_upper}, and \re{eta_goal_pre2}, we obtain \re{eta_goal_1}.
\end{proof}

In the same way, the following theorem holds.
\begin{theorem}
The relation
\bea{
\ep(\rho, A, \mE)^2
\aeqge 
 \Big(\f{\mathfrak{R}(\abs{\expt{[Y_S,A]}_\rho}-\sqrt{\mF_\rho(A)}\{T_\RM{LRB}^\pr+ \half \Dl_{H_P}  T_\RM{LRB}\})}
 {\sqrt{\mF^\RM{cost}}+\Dl_F^\pr+2T_\RM{LRB}^\pr}\Big)^2 \no\\
&\spa -\Big(\nor{\mR_2^{\ast\ast}}_1+2 \nor{\mR_1^{\ast\ast}}_1\nor{A} +  2\nor{A}^2 \Big) T_\RM{LRB} 
}
holds. Here, $Y_S \defe H_S-\Lm_P\dg(H_P)$ and $\Dl_F^\pr \defe \sqrt{\mF_\rho(H_S)}+\Dl_{H_P}$. 
$H_P$ is the Hamiltonian of the memory system. 
$\mR^{\ast\ast}$ is defined as $\ep(\rho, A, \mE) = \ep(\rho, A, \mE, \mR^{\ast\ast})$. 
We expanded $\mR^{\ast\ast}$ as $\mR^{\ast\ast}_0+\theta \mR^{\ast\ast}_1+\theta^2 \mR^{\ast\ast}_2+O(\theta^3)$. 
\end{theorem}

\subsection{Upper bound of $\nor{Z_{SE_1}}$} \la{s_Z}

In this subsection, we prove \re{Z_upper} under an assumption
\bea{
U_\tot \dg (H_{S^\pr}+H_{E^\pr}) U_\tot =H_S+H_E \la{katei_con}
}
where $U_\tot \defe e^{-iH_\tot t}$. 
\re{katei_con} leads to 
\bea{
\expt{ H_{S^\pr}+H_{E^\pr} }_{\rho_\tot(t)}=\expt{ H_S+H_E }_{\rho_\tot(0)} \ \ \any \rho_\tot(0).
}
The equation and
\bea{
\expt{ H_\tot }_{\rho_\tot(t)}=\expt{ H_\tot }_{\rho_\tot(0)} \ \ \any \rho_\tot(0)
}
lead to 
\bea{
\expt{ H_\RM{int}^\pr }_{\rho_\tot(t)}=\expt{ H_\RM{int} }_{\rho_\tot(0)} \ \ \any \rho_\tot(0),
}
which implies
\bea{
U_\tot \dg  H_\RM{int}^\pr U_\tot =H_\RM{int} .\la{con_int}
}
\re{con_int} leads to 
\bea{
Z_{SE_1} \aeq H_\RM{int}-U_T\dg H_\RM{int}^\pr U_T \no\\
\aeq  U_\tot \dg  H_\RM{int}^\pr U_\tot -U_T\dg H_\RM{int}^\pr U_T
}
where $U_T\defe e^{-iH_Tt}$. 
Here, we used $U_T\dg H_T U_T= H_T$, which implies
\bea{
&Z_{SE_1}-H_\RM{int}+U_T\dg H_\RM{int}^\pr U_T \no\\ 
\aeq -(H_S+H_\RM{int}+H_{E_1})+ U_T\dg (H_{S^\pr}+H_\RM{int}^\pr+H_{E_1^\pr})U_T = 0.
}
$Z_{SE_1}$ is calculated as
\bea{
&U_\tot \dg  H_\RM{int}^\pr U_\tot-U_T\dg H_\RM{int}^\pr U_T \no\\
\aeq \int_0^t ds \ \f{d}{ds}(e^{iH_\tot s}e^{iH_T(t-s)}H_\RM{int}^\pr e^{-iH_T(t-s)}e^{-iH_\tot s}) \no\\
\aeq \int_0^t ds \ \Big(e^{iH_\tot s}i(H_\tot-H_T)e^{iH_T(t-s)}H_\RM{int}^\pr e^{-iH_T(t-s)}e^{-iH_\tot s} \no\\
&\spa -e^{iH_\tot s}e^{iH_T(t-s)}H_\RM{int}^\pr e^{-iH_T(t-s)}i(H_\tot-H_T)e^{-iH_\tot s} \Big) \no\\
\aeq \int_0^t ds \ \Big(e^{iH_\tot s}i(H_{\p E_1^\pr}+H_{E_2})e^{iH_T(t-s)}H_\RM{int}^\pr e^{-iH_T(t-s)}e^{-iH_\tot s} \no\\
&\spa -e^{iH_\tot s}e^{iH_T(t-s)}H_\RM{int}^\pr e^{-iH_T(t-s)}i(H_{\p E_1^\pr}+H_{E_2})e^{-iH_\tot s} \Big) \no\\
\aeq -i\int_0^t ds \ e^{iH_\tot s}[e^{iH_T(t-s)}H_\RM{int}^\pr e^{-iH_T(t-s)},H_{\p E_1^\pr}+H_{E_2}]e^{-iH_\tot s} .
}
Because $H_{E_2}$ commutes with $H_T$ and $H_\RM{int}^\pr$, the above equation becomes
\bea{
Z_{SE_1} \aeq -i\int_0^t ds \ e^{iH_\tot s}[e^{iH_T(t-s)}H_\RM{int}^\pr e^{-iH_T(t-s)},H_{\p E_1^\pr}]e^{-iH_\tot s} .
}
Thus, we obtain
\bea{
\nor{Z_{SE_1}} \aeqle \int_0^t ds \ \nor{[e^{iH_T(t-s)}H_\RM{int}^\pr e^{-iH_T(t-s)},H_{\p E_1^\pr}]} .
}
Using the Lieb-Robinson bound \re{e_LRB}, we obtain
\bea{
\nor{Z_{SE_1}} \aeqle \int_0^t ds \  C\nor{H_\RM{int}^\pr} \cdot \nor{H_{\p E_1^\pr}} \cdot \abs{I^\pr}\abs{\p E_1^\pr}
e^{-\mu d(I^\pr,\p E_1^\pr)} \int_0^t ds \ (e^{v(t-s)}-1) \no\\
\aeq C \nor{H_\RM{int}^\pr}  \cdot \nor{H_{\p E_1^\pr}} \cdot 
\abs{I^\pr}\abs{\p E_1^\pr}e^{-\mu d(I^\pr, \p E_1^\pr)} \f{e^{vt}-1-vt}{v} \no\\
\aeq T_\RM{LRB}^\pr . \la{Z_up}
}

\subsection{Upper bound of $\sqrt{\mF^\RM{cost}}$} \la{s9}

In this subsection, we derive an upper bound of $\sqrt{\mF^\RM{cost}}$. 
We start from
\bea{
\mF^\RM{cost} \le 4\Big\{\tr(H_{E_1}^2\rho_{E_1})- [\tr(H_{E_1}\rho_{E_1})]^2\Big\} 
\le 4\tr(H_{E_1}^2\rho_{E_1}) \le 4\nor{H_{E_1}}^2. \la{mF_upper}
}
Using \re{H_E_1} and \re{mF_upper}, we obtain
\bea{
\sqrt{\mF^\RM{cost}}\aeqle 2\sum_{Z \subset \Lm_{E_1}\backslash \{0\}} \nor{h_Z} \no\\
\aeqle 2\sum_{x \in  \Lm_{E_1}\backslash \{0\}}\sum_{\Lm_E \supset Z \ni x} \nor{h_Z} \no\\
\aeq 2\sum_{\Lm_E \supset Z \ni y} \nor{h_Z}
+2\sum_{x \in  \Lm_{E_1}\backslash \{0,y\}}\sum_{\Lm_E \supset Z \ni x} \nor{h_Z}.
}
Here, we fixed $y \in S^\pr$.
We introduce
\bea{
S_y(r) \defe \abs{\{x \in \Lm_E \vert  r-1 <d(x,y)/l_0 \le r  \}}
}
where $l_0$ is the unit of distance. 
We suppose that
\bea{
S_y(r) \le K^\pr r^{d-1} \ \ (r\ge 1).
}
We take a real number $R$ satisfying $d(x,y)/l_0 \le R$ for all $x \in \Lm_{E_1}\backslash \{0\}$. 
Then, we obtain
\bea{
\sqrt{\mF^\RM{cost}} \aeqle 2h_\RM{site}^{(y)} + 2h_\RM{site}^\RM{max} \sum_{x \in  \Lm_{E_1}\backslash \{0,y\}} \no\\
\aeqle  2h_\RM{site}^{(y)} + 2h_\RM{site}^\RM{max} \sum_{r=1}^R S_y(r) \no\\
\aeqle 2h_\RM{site}^{(y)} + 2h_\RM{site}^\RM{max} K^\pr \sum_{r=1}^R r^{d-1} \la{nor_H_E_1} \no\\
\aeq  2h_\RM{site}^{(0)} +2K^\pr h_\RM{site}^\RM{max} s(R) 
}
with
\bea{
h_\RM{site}^{(y)} \aeqd \sum_{\Lm_E \supset Z \ni y}\nor{h_Z} ,\\
h_\RM{site}^\RM{max} \aeqd \max_{ x\in \Lm_E }\sum_{\Lm_E \supset Z \ni x}\nor{h_Z}, \\
s(R) \aeqd \sum_{r=1}^R r^{d-1}.
}
$s(R)$ is upper bounded as
\bea{
s(R)\aeqle \int_1^{R+1} dr\ r^{d-1} \no\\
 \aeqle \f{(R+1)^d-1}{d}.
}
Thus, we obtain
\bea{
\sqrt{\mF^\RM{cost}} \le 2h_\RM{site}^{(y)} + \f{2K^\pr h_\RM{site}^\RM{max}}{d}[(R+1)^d-1] .\la{F_jyougen}
}

\subsection{$S^\pr=S$ case} \la{s10}

In this subsection, we consider $S^\pr=S$ case.
In this case, $Z_{SE_1}=0$ holds. 
We prove this under the assumption 
\bea{
U_\tot \dg (H_S+H_E) U_\tot =H_S+H_E,
}
which is equivalent to
\bea{
[H_\tot, H_S+H_E]=0. \la{c_1}
}
$Z_{SE_1}=0$ is equivalent to
\bea{
[H_T, H_S+H_{E_1}]=0. \la{c_2}
}
\re{c_1} and \re{c_2} are equivalent to
\bea{
[H_\RM{int}, H_S+H_E]=0
}
and 
\bea{
[H_\RM{int}, H_S+H_{E_1}]=0.
}
Then, we need to prove only
\bea{
[H_\RM{int}, H_E-H_{E_1}]=0.
}
The left-hand-side is given by
\bea{
[H_\RM{int}, H_E-H_{E_1}]=[H_\RM{int}, H_{\p E_1^\pr}+H_{E_2}]=[H_\RM{int}, H_{\p E_1^\pr}].
}
Because $d(I, \p E_1^\pr)>0$ where $I$ is the support of $H_\RM{int}$, $[H_\RM{int}, H_{\p E_1^\pr}]=0$ holds.
Thus, $Z_{SE_1}=0$ holds.

Because $Z_{SE_1}=0$, Theorem \ref{t_main} becomes the following theorem.
 
\begin{theorem} 
If $S^\pr=S$, the following relation holds:
\bea{
\eta(\rho, A, \mE)^2 \aeqge 
 \Big(\f{\mathfrak{R}(\abs{\expt{[Y_S^\pr,A]}_\rho}-\half \sqrt{\mF_\rho(A)} \Dl_{H_{S}}  T_\RM{LRB})}{\sqrt{\mF^\RM{cost}}+\Dl_F}\Big)^2 \no\\
&\spa-\Big(\nor{\mR_2^\ast}_1+2 \nor{\mR_1^\ast}_1\nor{A} +2 \nor{A}^2 \Big)  T_\RM{LRB} .
}
Here, $Y_S^\pr\defe H_S-\Lm_{S}\dg(H_S)$. 
$\Dl_F$ and $\mF^\RM{cost}$ are defined by
\bea{
\Dl_F \aeqd \sqrt{\mF_\rho(H_S)}+\Dl_{H_S} ,\\
\mF^\RM{cost}\aeqd \min \{\mF_{\rho_{E_1}}(H_{E_1})\Big\vert (\rho_{E_1}, H_{E_1}, U_T)\to \tl \Lm_S  \} .
}
Here, $ (\rho_{E_1}, H_{E_1}, U_T)\to \tl \Lm_S$ means 
\bea{
\tl \Lm_S(\bu) \aeq \tr_{E_1}[U_T\bu \otimes \rho_{E_1}U_T \dg] ,\\
H_S+H_{E_1} \aeq U_T\dg (H_{S}+H_{E_1})U_T 
}
hold.
\end{theorem}

\subsection{Derivation of main result} \la{s11}

We derive the main result. 
\re{eta_goal_1} and \re{F_jyougen} leads to 
\bea{
\eta(\rho, A, \mE)^2 
 \aeqge 
 \Big(\f{\mathfrak{R}(\abs{\expt{[Y_S^\pr,A]}_\rho}-\sqrt{\mF_\rho(A)} T_1)}{K(R_0+R_1t)^d+c^\pr+T_2}\Big)^2 -T_3 ,\la{matome}\\
 T_1 \aeqd T_\RM{LRB}^\pr+ \half \Dl_{H_{S^\pr}} T_\RM{LRB} ,\\
 T_2 \aeqd 2T_\RM{LRB}^\pr ,\\
 T_3 \aeqd \Big(\nor{\mR_2^\ast}_1+2 \nor{\mR_1^\ast}_1\nor{A} +  2\nor{A}^2 \Big) T_\RM{LRB} ,\\
 c^\pr \aeqd 2h_\RM{site}^{(y)} -K+\Dl_F ,\\
 R_0 \aeqd \bar{R}_0+1,\\
 K \aeqd  \f{2K^\pr h_\RM{site}^\RM{max}}{d}.
}
Here, we set 
\bea{
R \aeq \bar{R}_0+R_1t \com -\mu l_0R_1+v<0 .
}
Using $\eta(\rho, A, \mE, \mR_{B}) + \sqrt{T_3}\ge \sqrt{\eta(\rho, A, \mE, \mR_{B})^2+T_3}$, \re{matome} leads to 
\bea{
\eta(\rho, A, \mE) \ge \f{\abs{\expt{[Y_S^\pr,A]}_\rho}-\sqrt{\mF_\rho(A)} T_1}{K(R_0+R_1t)^d+c^\pr+T_2} -\sqrt{T_3} . \la{matome_2}
}

We evaluate upper bounds of $T_i(i=1,2,3)$. 
First, 
\bea{
\nor{H_{\p E_1^\pr}} \aeqle \sum_{Z \subset \Lm_{\p E_1^\pr}} \nor{h_Z} \no\\
\aeqle \sum_{x \in  \Lm_{\p E_1^\pr}}\sum_{ Z \ni x} \nor{h_Z} \no\\
\aeqle \abs{\p E_1^\pr} h_\RM{site}^\RM{max}
}
holds. 
Here, $\Lm_{\p E_1^\pr}$ is the set of all sites in $\p E_1^\pr$. 
Then, we obtain
\bea{
\nor{H_{\p E_1^\pr}}\cdot \abs{\p E_1^\pr} \le \abs{\p E_1^\pr}^2 h_\RM{site}^\RM{max}.
}
We introduce $\tl R \defe d(I^\pr, \p E_1^\pr)/l_0 $ and suppose that there exist $K_\p>0$, $D>0$, and $k_0\ge 0$ such that
\bea{
\abs{\p E_1^\pr} \aeqle K_\p (\tl R+k_0)^{D}.
}
Then, 
\bea{
\nor{H_{\p E_1^\pr}}\cdot \abs{\p E_1^\pr} \le  K_\p^2 h_\RM{site}^\RM{max}(\tl R+k_0)^{2D} 
}
holds. This relation leads to 
\bea{
T_\RM{LRB} \aeqle \abs{S^\pr} C^\pr(\tl R+k_0)^{2D} e^{-\mu l_0\tl R}(e^{vt}-vt -1) ,\\
T_\RM{LRB}^\pr \aeqle  \nor{H_\RM{int}^\pr}\cdot\abs{I^\pr} C^\pr (\tl R+k_0)^{2D}  e^{-\mu l_0\tl R} (e^{vt}-vt-1) 
}
with $C^\pr \defe \f{C}{v}K_\p^2 h_\RM{site}^\RM{max}$. 
Here, we used $d(S^\pr, \p E_1^\pr) \ge d(I^\pr, \p E_1^\pr)$. 
$R$ can be taken
\bea{
R = \tl R+k_1
}
where $k_1$ is a constant. 
Thus, we can take
\bea{
\tl R \aeq \tl R_0 +R_1t 
}
and $\bar{R}_0=\tl R_0+k_1$. 
Introducing
\bea{
T_0\aeqd (\tl R_0 +k_0+R_1t)^{2D} e^{-\mu l_0 (\tl R_0 +R_1t)}(e^{vt}-vt -1),
}
we obtain
\bea{
T_\RM{LRB} \aeqle \abs{S^\pr} C^\pr T_0 ,\\
T_\RM{LRB}^\pr \aeqle \nor{H_\RM{int}^\pr}\cdot\abs{I^\pr} C^\pr  T_0
}
and 
\bea{
T_i \le T_i^\RM{max} 
}
with 
\bea{
 T_1^\RM{max}  \aeqd C^\pr T_0(\nor{H_\RM{int}^\pr}\cdot\abs{I^\pr}  +\half \Dl_{H_{S^\pr}} \abs{S^\pr}) ,\\
 T_2^\RM{max} \aeqd 2 C^\pr T_0 \nor{H_\RM{int}^\pr}\cdot\abs{I^\pr}   ,\\
 T_3^\RM{max} \aeqd  C^\pr T_0\Big(\nor{\mR_2^\ast}_1+2 \nor{\mR_1^\ast}_1\nor{A} +  2\nor{A}^2 \Big) \abs{S^\pr}.
}
Thus, \re{matome_2} leads to
\bea{
\eta(\rho, A, \mE) \aeqge \f{\abs{\expt{[Y_S^\pr,A]}_\rho}-\sqrt{\mF_\rho(A)} T_1^\RM{max}}{K(R_0+R_1t)^d+c^\pr+T_2^\RM{max}} -\sqrt{T_3^\RM{max}} \no\\
\aeq \f{\abs{\expt{[Y_S^\pr,A]}_\rho}-\sqrt{\mF_\rho(A)} T_1^\RM{max}-\sqrt{T_3^\RM{max}}[K(R_0+R_1t)^d+c^\pr+T_2^\RM{max}] }{K(R_0+R_1t)^d+c^\pr+T_2^\RM{max}} .
}
If $\abs{\expt{[Y_S^\pr,A]}_\rho}>0$, we can take $R_0$ such that
\bea{
\abs{\expt{[Y_S^\pr,A]}_\rho}-\sqrt{\mF_\rho(A)} T_1^\RM{max}-\sqrt{T_3^\RM{max}}[K(R_0+R_1t)^d+c^\pr+T_2^\RM{max}] \ge \half \abs{\expt{[Y_S^\pr,A]}_\rho}
}
for all $t\ge 0$. 
Then, we obtain
\bea{
\eta(\rho, A, \mE) \ge \half \f{\abs{\expt{[Y_S^\pr,A]}_\rho}}{K(R_0+R_1t)^d+c}  
}
with $c \defe c^\pr+\max_{t\ge 0}T_2^\RM{max}$. 
This is \re{e_main_d}. 

Similarly, we obtain \re{e_main_e} with
\bea{
c_1 \aeqd c_1^\pr+\max_{t\ge 0} T_2^\RM{max} \com c_1^\pr \defe 2h_\RM{site}^{(y)} -K+\Dl_F^\pr.
}

\subsection{Yanase condition} \la{s_Yanase}

In this subsection, we suppose 
\bea{
\Lm_P(\bu) \aeq \sum_m \tr_{S^\pr}[\mE_m(\rho)] \ke{m}\bra{m}_P 
}
where $\{\mE_m\}$ are completely positive maps from $S$ to $S^\pr$. 
$\{\ke{m}_P \}$ is the orthonormal basis on $P$. 
We also suppose the Yanase condition
\bea{
[H_P, \ke{m}\bra{m}_P ] \aeq 0 \ \ \mbox{for all }m.
}
Here, $H_P$ is the Hamiltonian of the memory system. 

We introduce another system $P^\pr$ whose Hilbert space is the same as that of $P$. 
We consider a CPTP map $\mV$ from $P$ to $P^\pr$ defined by \cite{ET2023}
\bea{
\mV(\bu) \aeqd \tr_P[V \bu \otimes \ke{0}\bra{0}_{P^\pr} V\dg] ,\\
V \aeqd \ke{0}\bra{0}_P \otimes 1_{P^\pr}+
\sum_{m(\ne 0)}\ke{m}\bra{m}_P \otimes \Big( \ke{m}\bra{0}_{P^\pr}+\ke{0}\bra{m}_{P^\pr} +\sum_{i(\ne 0, m)} \ke{i}\bra{i}_{P^\pr}\Big).
} 
Here, $\{\ke{m}_{P^\pr}\}$ is an orthonormal basis on $P^\pr$. 
$V$ is unitary. 
Because of
\bea{
\mV(\ke{m}\bra{m^\pr}_P) \aeq \dl_{mm^\pr} \ke{m}\bra{m}_{P^\pr}, 
}
we obtain
\bea{
\mV\Big(\sum_m p_m\ke{m}\bra{m}_P\Big) \aeq \sum_m p_m\ke{m}\bra{m}_{P^\pr}.
}
Thus, $\mV \circ \Lm_P(\rho)$ is the same state with $\Lm_P(\rho)$. 
This leads to 
\bea{
\ep(\rho, A, \mE) \aeq \ep(\rho, A, \mV \circ \Lm_P).
}

$V$ satisfies $[V, H_P+h1_{P^\pr}]=0$ for an arbitrary real number $h$ since the Yanase condition. 
We put $H_{P^\pr} \defe h1_{P^\pr}$. 
We introduce
\bea{
\Lm_{P^\pr}^\pr \aeqd \mV \circ \Lm_P ,\\
\tl \Lm_{P^\pr}^\pr \aeqd \mV \circ \tl \Lm_P.
}
Here, $\tl \Lm_P$ is $\tl \Lm_{S^\pr}$ for $S^\pr=P$.
Using \re{IKS} and \re{nor_1_Ga=1}, we obtain
\bea{
\nor{\Lm_{P^\pr}^\pr(X)-\tl \Lm_{P^\pr}^\pr(X)}_1 \le \nor{X}_1 T_\RM{LRB}. 
}
If 
\bea{
\tl \Lm_{P}(\bu) \aeq \tr_{E_1^\pr}[U_T\bu \otimes \rho_{E_1} U_T \dg] , \\
Z_{SE_1} \aeq  U_T\dg (H_{P}+H_{E_1^\pr})U_T-(H_S+H_{E_1})
}
hold, we obtain
\bea{
\tl \Lm_{P^\pr}^\pr(\bu) \aeq \tr_{E_1^\pr P}[VU_T\bu \otimes \rho_{E_1} \otimes \ke{0}\bra{0}_{P^\pr}  U_T \dg V\dg] , \\
Z_{SE_1} \aeq  U_T\dg V\dg(H_{P}+H_{P^\pr}+H_{E_1^\pr})VU_T-(H_S+H_{P^\pr}+H_{E_1}).
}
Using the same process as obtaining \re{S160}, we obtain
\bea{
\ep(\rho, A, \tl \Lm_{P^\pr}^\pr, \mR^{\pr\ast}) \ge \f{\abs{\expt{[\calY_S,A]}_\rho}-\sqrt{\mF_\rho(A)}\{\nor{Z_{SE_1}}+\half \Dl_{H_{P^\pr}} T_\RM{LRB}\}}{\sqrt{\mF_{\rho_{E_1}\otimes \ke{0}\bra{0}_{P^\pr} }(H_{E_1}+H_{P^\pr})}+\tl \Dl_F+\Dl_1}  
}
with
\bea{
\tl \Dl_F \aeqd \sqrt{\mF_\rho(H_S)}+2\sqrt{V_{\tl \Lm_{P^\pr}(\rho)}(H_{P^\pr})} \le \Dl_{H_S} ,\\
\Dl_1 \aeqd \max_{\rho_S}\sqrt{\mF_{\rho_S \otimes \rho_{E_1} \otimes \ke{0}\bra{0}_{P^\pr} }(Z_{SE_1})} \le 2\nor{Z_{SE_1}}.  
}
Here, 
\bea{
\calY_S \defe H_S-(\Lm_{P^\pr}^\pr)\dg(H_{P^\pr})=H_S-h1_S. \la{cal_Y_S}
}
$\mR^{\pr\ast}$ is defined by $\min_{\mR^\pr} \ep(\rho, A, \Lm_{P^\pr}^\pr, \mR^{\pr})=\ep(\rho, A, \Lm_{P^\pr}^\pr, \mR^{\pr\ast})$. 
Using \re{cal_Y_S} and
\bea{
\mF_{\rho_{E_1}\otimes \ke{0}\bra{0}_{P^\pr} }(H_{E_1}+H_{P^\pr}) \aeq \mF_{\rho_{E_1}}(H_{E_1}), 
}
we obtain 
\bea{
\ep(\rho, A, \tl \Lm_{P^\pr}^\pr, \mR^{\pr\ast}) \ge \f{\abs{\expt{[H_S,A]}_\rho}-\sqrt{\mF_\rho(A)}\nor{Z_{SE_1}}}{\sqrt{\mF^\RM{cost}}+\Dl_{H_S}+2\nor{Z_{SE_1}}}  . \la{S_ep_Ppr}
}
We used $\Dl_{H_{P^\pr}}=0$. 
\re{S_ep_Ppr} leads to 
\bea{
\ep(\rho, A, \mE) \ge \half \f{\abs{\expt{[H_S, A]}_\rho}}{K(R_0+R_1t)^d+c_1}.  \la{e_main_Yanase}
}

\subsection{Gate-fidelity error} \la{s12}

For two CPTP maps $\mN_1$ and $\mN_2$ from system $K$ to $K^\pr$, 
we introduce 
\bea{
\mT(\mN_1, \mN_2) \aeqd \max_\rho \mT(\mN_1(\rho), \mN_2(\rho)) ,\\
D_F(\mN_1, \mN_2) \aeqd  \max_\rho D_F(\mN_1(\rho), \mN_2(\rho))
}
where $\mT(\rho, \sig) \defe \nor{\rho-\sig}_1/2$ is the trace distance. 
The following lemma holds. 
\begin{lemma} 
We consider a CPTP map $\mN$ from system $K$ to $K$ and a unitary map $\mU$ on $K$. 
Then, 
\bea{
D_F(\mN, \mU) \ge \dl(\mN, \Om) \la{l_gfe_1}
}
holds for an arbitrary test ensemble $\Om$.
\end{lemma}

\begin{proof} We prove \re{l_gfe_1}:
\bea{
\dl(\mN, \Om)^2 \aeq \min_\mR \sum_k p_k D_F(\rho_k, \mR \circ \mN(\rho_k))^2 \no\\
\aeqle  \sum_k p_k D_F(\rho_k, \mU\dg \circ \mN(\rho_k))^2 \no\\
\aeq   \sum_k p_k D_F(\mU(\rho_k), \mN(\rho_k))^2 \no\\
\aeqle D_F(\mU, \mN)^2.
}
\end{proof}

Using the SIQ theorem, we obtain the following theorem.
\begin{theorem} \la{t_gfe_pre}
We consider $S^\pr=S$ case and a unitary map $\mU$ on $S$. 
Then, the relation
\bea{
D_F(\Lm_{S}, \mU) \ge \f{1}{\sqrt{\mF^\RM{cost}}+\Dl}\Big[C_U-\Dl_{H_S}\big(\half T_\RM{LRB}+2D_F(\Lm_S, \mU) \big) \Big]
-\sqrt{T_\RM{LRB}}  \la{l_gft_0}
}
holds with
\bea{
C_U \aeqd \abs{\bra{\psi_+}[H_S-\mU\dg(H_S)]\ke{\psi_-}} .\la{def_C_U}
}
Here, $\ke{\psi_+}$ and $\ke{\psi_-}$ are states on $S$ that are orthogonal to each other. 
$\Dl_{H_S}$ is the difference between the largest and smallest eigenvalues of $H_S$. 
$\Dl \defe 2\Dl_{H_S}$ and $\mF^\RM{cost}$ is given by \re{def_F_cost}.
\end{theorem} 

\begin{proof}

We prove \re{l_gft_0}. 
First, 
\bea{
D_F(\Lm_{S}, \mU) \ge D_F(\tl \Lm_{S}, \mU)- D_F(\tl \Lm_{S}, \Lm_{S})
}
holds. 
Using the formula \cite{Nielsen2002quantum_computation}
\bea{
1-F(\rho, \sig) \le \mT(\rho, \sig) \le D_F(\rho, \sig), \la{F_mT_D_F}
}
we obtain
\bea{
D_F(\rho, \sig)^2 \le \mT(\rho, \sig)[1+F(\rho, \sig)] \le 2\mT(\rho, \sig).
}
Using this and $2\mT(\tl \Lm_{S}, \Lm_{S})\le T_\RM{LRB}$ derived from  \re{IKS}, we obtain
\bea{
D_F(\tl \Lm_{S}, \Lm_{S}) \le \sqrt{T_\RM{LRB}}
}
and 
\bea{
D_F(\Lm_{S}, \mU) \ge D_F(\tl \Lm_{S}, \mU)-  \sqrt{T_\RM{LRB}} .
}
Using this and \re{l_gfe_1}, we obtain
\bea{
D_F(\Lm_{S}, \mU) \ge \dl(\tl \Lm_{S}, \Om)- \sqrt{T_\RM{LRB}} .\la{gfe1}
}

Because $Z_{SE_1}=0$ as we showed in \res{s10}, we can use the SIQ Theorem \ref{theo_0}.
For $\Om_1\defe \{(1/2, \psi_+), (1/2, \psi_-)\}$, we obtain
\bea{
\dl(\tl \Lm_{S}, \Om_1) \ge \f{\sqrt{2}\mC}{\sqrt{\mF^\RM{cost}}+\tl \Dl} . \la{gfe2}
}
Here, 
\bea{
\tl \Dl \aeqd \max_{\rho \in \RM{supp}(\psi_k)}\sqrt{\mF_{\rho\otimes \rho_{E_1}}(H_S-U_T\dg H_{E_1}U_T)} \no\\
\aeqle \max_{\rho \in \RM{supp}(\psi_k)}\Big[\sqrt{\mF_\rho(H_S)}+2\sqrt{V_{\tl \Lm_{S}(\rho)}(H_S)} \Big] \no\\
\aeqle 2\Dl_{H_S}=\Dl \la{Dl_vari}
}
with $\tl \Lm_{S}(\bu)=\tr_{E_1}[U_T\bu \otimes \rho_{E_1} U_T\dg]$ and $H_S+H_{E_1}=U_T\dg(H_S+H_{E_1})U_T$. 
We used \re{vari_X} in \re{Dl_vari}. 
$\mC^2$ is given by
\bea{
\mC^2 \aeq \half \abs{\bra{\psi_+}[H_S-\tl \Lm_{S}\dg (H_S)]\ke{\psi_-}}^2.
}
Then, we obtain
\bea{
\sqrt{2} \mC \aeq \abs{\bra{\psi_+}[H_S-\tl \Lm_{S}\dg (H_S)]\ke{\psi_-}} \no\\
\aeqle \abs{\bra{\psi_+}[H_S-\mU\dg (H_S)]\ke{\psi_-}}- \abs{\bra{\psi_+}[\Lm_{S}\dg(H_S)-\tl \Lm_{S}\dg (H_S)]\ke{\psi_-}} \no\\
&\spa- \abs{\bra{\psi_+}[\Lm_{S}\dg(H_S)- \mU\dg (H_S)]\ke{\psi_-}} .\la{gfe3}
}
By putting $\om \defe \ke{\psi_-}\bra{\psi_+}$, we obtain
\bea{
\abs{\bra{\psi_+}[\Lm_{S}\dg(H_S)-\tl \Lm_{S}\dg (H_S)]\ke{\psi_-}} \aeq \abs{\tr_S(\om[\Lm_{S}\dg(H_S)-\tl \Lm_{S}\dg (H_S)] )} \no\\
\aeq \abs{\tr_S(H_S[\Lm_{S}(\om)-\tl \Lm_{S} (\om)] )}  \no\\
\aeq \abs{\tr_S((H_S-x1_S)[\Lm_{S}(\om)-\tl \Lm_{S} (\om)] )}\no\\
\aeqle \nor{H_S-x1_S}\cdot\nor{\Lm_{S}(\om)-\tl \Lm_{S} (\om)}_1 \no\\
\aeqle \half \Dl_{H_S}\nor{\om}_1 T_\RM{LRB} \no\\
\aeq \half \Dl_{H_S} T_\RM{LRB}
}
using \re{IKS}. 
Here, $x$ is the average of the largest and smallest eigenvalues of $H_S$.
Using \re{super_norm}, we obtain
\bea{
 \abs{\bra{\psi_+}[\Lm_{S}\dg(H_S)- \mU\dg (H_S)]\ke{\psi_-}} \aeqle  \half \Dl_{H_S}\nor{\Lm_{S}(\om)- \mU (\om)}_1 \no\\
 \aeqle 2\Dl_{H_S}\mT(\Lm_{S} ,\mU) \no\\
 \aeqle 2\Dl_{H_S}D_F(\Lm_{S} ,\mU).
}
We used \re{F_mT_D_F}. 
Thus, \re{gfe3} becomes 
\bea{
\sqrt{2} \mC\aeqle C_U- \Dl_{H_S} \Big[\half  T_\RM{LRB} +2D_F(\Lm_{S} ,\mU) \Big] .\la{gfe4}
}
\re{gfe1}, \re{gfe2}, and \re{gfe4} lead to \re{l_gft_0}. 
\end{proof}

\begin{lemma} 
For two CPTP maps $\mN_1$ and $\mN_2$ from system $K$ to $K^\pr$ and
two states $\ke{\psi_+}$ and $\ke{\psi_-}$ on $K$ that are orthogonal to each other, 
the relation
\bea{
\nor{\mN_1(\ke{\psi_-}\bra{\psi_+})-\mN_2(\ke{\psi_-}\bra{\psi_+})  }_1 \le 4\mT(\mN_1, \mN_2) \la{super_norm}
}
holds. 
\end{lemma}

\begin{proof}
An arbitrary linear operator $X$ is decomposed as 
\bea{
X \aeq \sum_{a=1}^4 \ep_a X_a ,\\
(X_1, X_2, X_3, X_4) \aeqd (X_R^+, X_R^-, X_I^+, X_I^-) ,\\
(\ep_1, \ep_2, \ep_3, \ep_4) \aeqd (1, -1,i, -i) .
}
Here, $X_R \defe (X+X\dg)/2$ and $X_I\defe (X-X\dg)/(2i)$. 
$Y^+$ and $Y^-$ denote positive and negative parts of $Y$ ($Y=Y^+-Y^-$). 
Thus, we obtain 
\bea{
\nor{\mN_1(X)-\mN_2(X) }_1 \aeqle \sum_{a=1}^4 \nor{\mN_1(X_a)-\mN_2(X_a) }_1 \no\\
\aeq \sum_{a=1}^4\tr(X_a) \nor{\mN_1(\rho_a)-\mN_2(\rho_a) }_1.
}
Here, $\rho_a$ are defined as $X_a=\tr(X_a)\rho_a$. 
Because $\rho_a$ are density operators, $\nor{\mN_1(\rho_a)-\mN_2(\rho_a) }_1 \le 2\mT(\mN_1, \mN_2)$ holds. 
Thus, we obtain
\bea{
\nor{\mN_1(X)-\mN_2(X) }_1 \aeqle 2\mT(\mN_1, \mN_2)\sum_{a=1}^4\tr(X_a).
}
Because $\tr(\om_a)=1/2$ $(a=1,2,3,4)$ for $\om \defe \ke{\psi_-}\bra{\psi_+} $, we obtain 
\bea{
\nor{\mN_1(\om)-\mN_2(\om)  }_1 \le 4\mT(\mN_1, \mN_2).
}
\end{proof}

\noindent{\textbf{Remark.}}  For two arbitrary CPTP maps $\mN_1$ and $\mN_2$ from system $K$ to $K^\pr$ and an arbitrary linear operator $X$, 
\bea{
\half \nor{\mN_1(X)-\mN_2(X) }_1 \aeqle \nor{X}_1\mT(\mN_1, \mN_2)
}
does not hold (Proposition 3 in \cite{Watrous05}).

Theorem \ref{t_gfe_pre} leads to Theorem \ref{theo_gfe}.

\begin{proof} We prove \re{gfe_goal}. 
By putting $k \defe 1/(\sqrt{\mF^\RM{cost}}+\Dl)$, \re{l_gft_0} leads to 
\bea{
(1+2k\Dl_{H_S})D_F(\Lm_{S}, \mU)  \aeqge k(C_U-\half \Dl_{H_S}T_\RM{LRB})-\sqrt{T_\RM{LRB}} \no\\
D_F(\Lm_{S}, \mU) \aeqge \f{k}{1+2k\Dl_{H_S}}\Big(C_U-\half \Dl_{H_S}T_\RM{LRB}-\Big[\sqrt{\mF^\RM{cost}}+\Dl\Big]\sqrt{T_\RM{LRB}}\Big) .
}
\re{F_jyougen} leads to 
\bea{
\sqrt{\mF^\RM{cost}} \le F_\RM{max} \defe 2h_\RM{site}^{(y)} + K[(R_0^\pr+R_1 t)^d-1] .
}
Here, we put $R=R_0^\pr-1+R_1 t$. 
By putting  $k_\RM{\min}\defe 1/( F_\RM{max} + \Dl)$, we obtain
\bea{
D_F(\Lm_{S}, \mU) \aeqge \f{k_\RM{\min}}{1+2k_\RM{\min}\Dl_{H_S}}\Big(C_U-\half \Dl_{H_S}T_\RM{LRB}-\big[F_\RM{max} +\Dl\big]\sqrt{T_\RM{LRB}}\Big). \la{l_gft_1}
}
If $C_U=0$, the above relation is reduced to $D_F(\Lm_{S}, \mU) \ge 0$. 
Then, we suppose $C_U > 0$. 
We suppose $-\mu l_0R_1+v<0$. 
Then, similarly to \res{s11}, we can take $R_0^\pr$ such that 
\bea{
C_U-\half \Dl_{H_S}T_\RM{LRB}-\big[F_\RM{max} +\Dl\big]\sqrt{T_\RM{LRB}} \ge \half C_U \com 2k_\RM{\min}\Dl_{H_S} \le 1
}
for all $t \ge 0$. 
Thus, we obtain
\bea{
D_F(\Lm_{S}, \mU) \aeqge \f{1}{4} \f{C_U}{K(R_0^\pr+R_1 t)^d-K+\Dl+2h_\RM{site}^{(y)}} = \f{1}{4} \f{C_U}{K(R_0^\pr+R_1 t)^d+K_1}
}
with $K_1 \defe -K+\Dl+2h_\RM{site}^{(y)}$. 

\end{proof}

\bibliographystyle{quantum}
\bibliography{tajima_SIQ}

\begin{thebibliography}{10}

\bibitem{PhysRevApplied.7.054020}
T.~Walter, P.~Kurpiers, S.~Gasparinetti, P.~Magnard,
  A.~Poto\ifmmode~\check{c}\else \v{c}\fi{}nik, Y.~Salath\'e, M.~Pechal,
  M.~Mondal, M.~Oppliger, C.~Eichler, and A.~Wallraff.
\newblock ``Rapid high-fidelity single-shot dispersive readout of
  superconducting qubits''.
\newblock \href{https://dx.doi.org/10.1103/PhysRevApplied.7.054020}{Phys. Rev.
  Appl. {\bf 7}, 054020}~(2017).

\bibitem{PRXQuantum.5.010307}
Yoshiki Sunada, Kenshi Yuki, Zhiling Wang, Takeaki Miyamura, Jesper Ilves,
  Kohei Matsuura, Peter~A. Spring, Shuhei Tamate, Shingo Kono, and Yasunobu
  Nakamura.
\newblock ``Photon-noise-tolerant dispersive readout of a superconducting qubit
  using a nonlinear purcell filter''.
\newblock \href{https://dx.doi.org/10.1103/PRXQuantum.5.010307}{PRX Quantum
  {\bf 5}, 010307}~(2024).

\bibitem{GABRIELLI1982201}
C.~Gabrielli, F.~Huet, M.~Keddam, and J.F. Lizee.
\newblock ``Measurement time versus accuracy trade-off analyzed for
  electrochemical impedance measurements by means of sine, white noise and step
  signals''.
\newblock
  \href{https://dx.doi.org/https://doi.org/10.1016/0022-0728(82)87141-6}{Journal
  of Electroanalytical Chemistry and Interfacial Electrochemistry {\bf 138},
  201--208}~(1982).

\bibitem{Bohr28}
Niels Bohr.
\newblock ``Das quantenpostulat und die neuere entwicklung der atomistik''.
\newblock
  \href{https://dx.doi.org/10.1007/978-3-642-64946-2_3}{Naturwissenschaften
  {\bf 16}, 245--257}~(1928).

\bibitem{Heisenberg27}
W.~Heisenberg.
\newblock ``{\"U}ber den anschaulichen inhalt der quantentheoretischen
  kinematik und mechanik''.
\newblock \href{https://dx.doi.org/10.1007/BF01397280}{Z. Phys. {\bf 43},
  172--198}~(1927).

\bibitem{Landau31}
Lev Landau and Rudolf Peierls.
\newblock ``Erweiterung des unbestimmtheitsprinzips f{\"u}r die relativistische
  quantentheorie''.
\newblock Zeitschrift f{\"u}r Physik {\bf 69}, 56--69~(1931).

\bibitem{Mandelstam45}
Leonid Mandelstam and IG~Tamm.
\newblock ``The uncertainty relation between energy and time in
  non-relativistic quantum mechanics''.
\newblock In Selected papers.
\newblock Pages 115--123.
\newblock Springer~(1991).

\bibitem{Aharonov61}
Y.~Aharonov and D.~Bohm.
\newblock ``Time in the quantum theory and the uncertainty relation for time
  and energy''.
\newblock \href{https://dx.doi.org/10.1103/PhysRev.122.1649}{Phys. Rev. {\bf
  122}, 1649--1658}~(1961).

\bibitem{Ban93}
Masashi Ban.
\newblock ``Relative-state formulation of quantum systems''.
\newblock \href{https://dx.doi.org/10.1103/PhysRevA.48.3452}{Phys. Rev. A {\bf
  48}, 3452--3465}~(1993).

\bibitem{Busch08}
Paul Busch.
\newblock ``The time--energy uncertainty relation''.
\newblock \href{https://dx.doi.org/10.1007/978-3-540-73473-4_3}{Pages 73--105}.
\newblock Springer Berlin Heidelberg. Berlin, Heidelberg~(2008).

\bibitem{Sagawa16}
Takahiro Sagawa and Masahito Ueda.
\newblock ``Quantum measurement and quantum control (in japanese)''.
\newblock SAIENSU-SHA Co.,Ltd. ~(2016).

\bibitem{Lieb72}
Elliott~H Lieb and Derek~W Robinson.
\newblock ``The finite group velocity of quantum spin systems''.
\newblock \href{https://dx.doi.org/10.1007/BF01645779}{Communications in
  mathematical physics {\bf 28}, 251--257}~(1972).

\bibitem{Nachtergaele06}
Bruno Nachtergaele, Yoshiko Ogata, and Robert Sims.
\newblock ``Propagation of correlations in quantum lattice systems''.
\newblock \href{https://dx.doi.org/10.1007/s10955-006-9143-6}{Journal of
  statistical physics {\bf 124}, 1--13}~(2006).

\bibitem{Arthurs65}
E~Arthurs and John~Larry Kelly~jr.
\newblock ``B.s.t.j. briefs: On the simultaneous measurement of a pair of
  conjugate observables''.
\newblock \href{https://dx.doi.org/10.1002/j.1538-7305.1965.tb01684.x}{Bell
  Syst. Tech. J. {\bf 44}, 725--729}~(1965).

\bibitem{Yamamoto86}
Y.~Yamamoto and H.~A. Haus.
\newblock ``Preparation, measurement and information capacity of optical
  quantum states''.
\newblock \href{https://dx.doi.org/10.1103/RevModPhys.58.1001}{Rev. Mod. Phys.
  {\bf 58}, 1001--1020}~(1986).

\bibitem{Arthurs88}
E.~Arthurs and M.~S. Goodman.
\newblock ``Quantum correlations: A generalized heisenberg uncertainty
  relation''.
\newblock \href{https://dx.doi.org/10.1103/PhysRevLett.60.2447}{Phys. Rev.
  Lett. {\bf 60}, 2447--2449}~(1988).

\bibitem{Ishikawa91}
Shiro Ishikawa.
\newblock ``Uncertainty relations in simultaneous measurements for arbitrary
  observables''.
\newblock \href{https://dx.doi.org/10.1016/0034-4877(91)90046-P}{Rep. Math.
  Phys. {\bf 29}, 257--273}~(1991).

\bibitem{Ozawa91}
Masanao Ozawa.
\newblock ``Quantum limits of measurements and uncertainty principle''.
\newblock In Cherif Bendjaballah, Osamu Hirota, and Serge Reynaud, editors,
  Quantum Aspects of Optical Communications.
\newblock
  \href{https://dx.doi.org/https://doi.org/10.1007/3-540-53862-3_161}{Pages
  1--17}.
\newblock Heidelberg, Berlin~(1991). Springer.

\bibitem{Ozawa03}
Masanao Ozawa.
\newblock ``Universally valid reformulation of the heisenberg uncertainty
  principle on noise and disturbance in measurement''.
\newblock \href{https://dx.doi.org/10.1103/PhysRevA.67.042105}{Phys. Rev. A
  {\bf 67}, 042105}~(2003).

\bibitem{Ozawa04}
Masanao Ozawa.
\newblock ``Uncertainty relations for noise and disturbance in generalized
  quantum measurements''.
\newblock
  \href{https://dx.doi.org/https://doi.org/10.1016/j.aop.2003.12.012}{Ann.
  Phys. (N.Y.) {\bf 311}, 350--416}~(2004).

\bibitem{Ozawa19}
Masanao Ozawa.
\newblock ``Soundness and completeness of quantum root-mean-square errors''.
\newblock \href{https://dx.doi.org/10.1038/s41534-018-0113-z}{npj Quantum Inf.
  {\bf 5}, 1}~(2019).

\bibitem{Ozawa21}
Masanao Ozawa.
\newblock ``Quantum disturbance without state change: Soundness and locality of
  disturbance measures''~(2021).

\bibitem{Watanabe11PRA}
Yu~Watanabe, Takahiro Sagawa, and Masahito Ueda.
\newblock ``Uncertainty relation revisited from quantum estimation theory''.
\newblock \href{https://dx.doi.org/10.1103/PhysRevA.84.042121}{Phys. Rev. A
  {\bf 84}, 042121}~(2011).

\bibitem{Watanabe11ARX}
Yu~Watanabe and Masahito Ueda.
\newblock ``Quantum estimation theory of error and disturbance in quantum
  measurement''~(2011).

\bibitem{Lee20ARX}
Jaeha Lee and Izumi Tsutsui.
\newblock ``Geometric formulation of universally valid uncertainty relation for
  error''~(2020).

\bibitem{Lee20ENT}
Jaeha Lee and Izumi Tsutsui.
\newblock ``Uncertainty relation for errors focusing on general povm
  measurements with an example of two-state quantum systems''.
\newblock \href{https://dx.doi.org/10.3390/e22111222}{Entropy {\bf 22},
  1222}~(2020).

\bibitem{Lee22}
Jaeha Lee.
\newblock ``A universal formulation of uncertainty relation for
  error--disturbance and local representability of quantum
  observables''~(2022).

\bibitem{Busch13}
Paul Busch, Pekka Lahti, and Reinhard~F. Werner.
\newblock ``Proof of heisenberg's error-disturbance relation''.
\newblock \href{https://dx.doi.org/10.1103/PhysRevLett.111.160405}{Phys. Rev.
  Lett. {\bf 111}, 160405}~(2013).

\bibitem{Busch14PRA}
Paul Busch, Pekka Lahti, and Reinhard~F. Werner.
\newblock ``Heisenberg uncertainty for qubit measurements''.
\newblock \href{https://dx.doi.org/10.1103/PhysRevA.89.012129}{Phys. Rev. A
  {\bf 89}, 012129}~(2014).

\bibitem{Busch14RMP}
Paul Busch, Pekka Lahti, and Reinhard~F. Werner.
\newblock ``Colloquium: Quantum root-mean-square error and measurement
  uncertainty relations''.
\newblock \href{https://dx.doi.org/10.1103/RevModPhys.86.1261}{Rev. Mod. Phys.
  {\bf 86}, 1261--1281}~(2014).

\bibitem{ET2023}
Haruki Emori and Hiroyasu Tajima.
\newblock ``Error and disturbance as irreversibility with applications: Unified
  definition, wigner--araki--yanase theorem and out-of-time-order
  correlator''~(2023).
\newblock  \href{http://arxiv.org/abs/2309.14172}{arXiv:2309.14172}.

\bibitem{Iyoda17}
Eiki Iyoda, Kazuya Kaneko, and Takahiro Sagawa.
\newblock ``Fluctuation theorem for many-body pure quantum states''.
\newblock \href{https://dx.doi.org/10.1103/PhysRevLett.119.100601}{Phys. Rev.
  Lett. {\bf 119}, 100601}~(2017).

\bibitem{Tajima22}
Hiroyasu Tajima, Ryuji Takagi, and Yui Kuramochi.
\newblock ``Universal trade-off structure between symmetry, irreversibility,
  and quantum coherence in quantum processes''~(2022).

\bibitem{Shiraishi_Tajima}
Naoto Shiraishi and Hiroyasu Tajima.
\newblock ``Efficiency versus speed in quantum heat engines: Rigorous
  constraint from lieb-robinson bound''.
\newblock \href{https://dx.doi.org/10.1103/PhysRevE.96.022138}{Phys. Rev. E
  {\bf 96}, 022138}~(2017).

\bibitem{Gour2008resource}
Gilad Gour and Robert~W Spekkens.
\newblock ``The resource theory of quantum reference frames: manipulations and
  monotones''.
\newblock \href{https://dx.doi.org/10.1088/1367-2630/10/3/033023}{New J. Phys.
  {\bf 10}, 033023}~(2008).

\bibitem{Marvian_2013}
Iman Marvian and Robert~W Spekkens.
\newblock ``The theory of manipulations of pure state asymmetry: I. basic
  tools, equivalence classes and single copy transformations''.
\newblock \href{https://dx.doi.org/10.1088/1367-2630/15/3/033001}{New Journal
  of Physics {\bf 15}, 033001}~(2013).

\bibitem{Marvian_thesis}
I.~Marvian.
\newblock ``Symmetry, asymmetry and quantum information''.
\newblock PhD thesis.
\newblock the University of Waterloo.
\newblock ~(2012).

\bibitem{skew_resource}
Chao Zhang, Benjamin Yadin, Zhi-Bo Hou, Huan Cao, Bi-Heng Liu, Yun-Feng Huang,
  Reevu Maity, Vlatko Vedral, Chuan-Feng Li, Guang-Can Guo, and Davide
  Girolami.
\newblock ``Detecting metrologically useful asymmetry and entanglement by a few
  local measurements''.
\newblock \href{https://dx.doi.org/10.1103/PhysRevA.96.042327}{Phys. Rev. A
  {\bf 96}, 042327}~(2017).

\bibitem{Takagi_skew}
R.~Takagi.
\newblock ``Skew informations from an operational view via resource theory of
  asymmetry''.
\newblock \href{https://dx.doi.org/10.1038/s41598-019-50279-w}{Sci. Rep. {\bf
  9}, 14562}~(2019).

\bibitem{Marvian_distillation}
I.~Marvian.
\newblock ``Coherence distillation machines are impossible in quantum
  thermodynamics''.
\newblock \href{https://dx.doi.org/10.1038/s41467-019-13846-3}{Nat. Commun.
  {\bf 11}, 25}~(2020).

\bibitem{YT}
Koji Yamaguchi and Hiroyasu Tajima.
\newblock ``Beyond i.i.d. in the resource theory of asymmetry: An
  information-spectrum approach for quantum fisher information''.
\newblock \href{https://dx.doi.org/10.1103/PhysRevLett.131.200203}{Phys. Rev.
  Lett. {\bf 131}, 200203}~(2023).

\bibitem{YT2}
Koji Yamaguchi and Hiroyasu Tajima.
\newblock ``Smooth {M}etric {A}djusted {S}kew {I}nformation {R}ates''.
\newblock \href{https://dx.doi.org/10.22331/q-2023-05-22-1012}{{Quantum} {\bf
  7}, 1012}~(2023).

\bibitem{Kudo_Tajima}
Daigo {Kudo} and Hiroyasu {Tajima}.
\newblock ``{Fisher information matrix as a resource measure in the resource
  theory of asymmetry with general connected-Lie-group symmetry}''~(2022).
\newblock  \href{http://arxiv.org/abs/2205.03245}{arXiv:2205.03245}.

\bibitem{Shitara_Tajima}
Tomohiro Shitara and Hiroyasu Tajima.
\newblock ``The i.i.d. state convertibility in the resource theory of asymmetry
  for finite groups and lie groups''~(2023).
\newblock  \href{http://arxiv.org/abs/2312.15758}{arXiv:2312.15758}.

\bibitem{Wigner1952}
E.~P. {Wigner}.
\newblock ``{Die Messung quantenmechanischer Operatoren}''.
\newblock \href{https://dx.doi.org/10.1007/BF01948686}{Zeitschrift fur Physik
  {\bf 133}, 101--108}~(1952).

\bibitem{Araki-Yanase1960}
Huzihiro Araki and Mutsuo~M. Yanase.
\newblock ``Measurement of quantum mechanical operators''.
\newblock \href{https://dx.doi.org/10.1103/PhysRev.120.622}{Phys. Rev. {\bf
  120}, 622--626}~(1960).

\bibitem{OzawaWAY}
Masanao Ozawa.
\newblock ``Conservation laws, uncertainty relations, and quantum limits of
  measurements''.
\newblock \href{https://dx.doi.org/10.1103/PhysRevLett.88.050402}{Phys. Rev.
  Lett. {\bf 88}, 050402}~(2002).

\bibitem{Korzekwa_thesis}
K.~Korezekwa.
\newblock ``Resource theory of asymmetry''.
\newblock PhD thesis.
\newblock Imperial College London.
\newblock ~(2013).

\bibitem{TN}
Hiroyasu {Tajima} and Hiroshi {Nagaoka}.
\newblock ``{Coherence-variance uncertainty relation and coherence cost for
  quantum measurement under conservation laws}''~(2019).
\newblock  \href{http://arxiv.org/abs/1909.02904}{arXiv:1909.02904}.

\bibitem{ozawaWAY_CNOT}
Masanao Ozawa.
\newblock ``Conservative quantum computing''.
\newblock \href{https://dx.doi.org/10.1103/PhysRevLett.89.057902}{Phys. Rev.
  Lett. {\bf 89}, 057902}~(2002).

\bibitem{Karasawa_2009}
Tokishiro Karasawa, Julio Gea-Banacloche, and Masanao Ozawa.
\newblock ``Gate fidelity of arbitrary single-qubit gates constrained by
  conservation laws''.
\newblock \href{https://dx.doi.org/10.1088/1751-8113/42/22/225303}{J. Phys. A:
  Math. Theor. {\bf 42}, 225303}~(2009).

\bibitem{TSS}
Hiroyasu Tajima, Naoto Shiraishi, and Keiji Saito.
\newblock ``Uncertainty relations in implementation of unitary operations''.
\newblock \href{https://dx.doi.org/10.1103/PhysRevLett.121.110403}{Phys. Rev.
  Lett. {\bf 121}, 110403}~(2018).

\bibitem{TSS2}
Hiroyasu Tajima, Naoto Shiraishi, and Keiji Saito.
\newblock ``Coherence cost for violating conservation laws''.
\newblock \href{https://dx.doi.org/10.1103/PhysRevResearch.2.043374}{Phys. Rev.
  Research {\bf 2}, 043374}~(2020).

\bibitem{TS}
Hiroyasu {Tajima} and Keiji {Saito}.
\newblock ``{Universal limitation of quantum information recovery: symmetry
  versus coherence}''~(2021).
\newblock  \href{http://arxiv.org/abs/2103.01876}{arXiv:2103.01876}.

\bibitem{Eastin-Knill}
Bryan Eastin and Emanuel Knill.
\newblock ``Restrictions on transversal encoded quantum gate sets''.
\newblock \href{https://dx.doi.org/10.1103/PhysRevLett.102.110502}{Phys. Rev.
  Lett. {\bf 102}, 110502}~(2009).

\bibitem{e-EKFaist}
Philippe Faist, Sepehr Nezami, Victor~V. Albert, Grant Salton, Fernando
  Pastawski, Patrick Hayden, and John Preskill.
\newblock ``Continuous symmetries and approximate quantum error correction''.
\newblock \href{https://dx.doi.org/10.1103/PhysRevX.10.041018}{Phys. Rev. X
  {\bf 10}, 041018}~(2020).

\bibitem{e-EKKubica}
Aleksander Kubica and Rafa\l{} Demkowicz-Dobrza\ifmmode~\acute{n}\else
  \'{n}\fi{}ski.
\newblock ``Using quantum metrological bounds in quantum error correction: A
  simple proof of the approximate eastin-knill theorem''.
\newblock \href{https://dx.doi.org/10.1103/PhysRevLett.126.150503}{Phys. Rev.
  Lett. {\bf 126}, 150503}~(2021).

\bibitem{Tajima_Takagi}
Hiroyasu Tajima and Ryuji Takagi.
\newblock ``Gibbs-preserving operations requiring infinite amount of quantum
  coherence''~(2024).
\newblock  \href{http://arxiv.org/abs/2404.03479}{arXiv:2404.03479}.

\bibitem{Watrous_text}
John Watrous.
\newblock ``The theory of quantum information''.
\newblock Cambridge university press. ~(2018).

\bibitem{Frowis_unc}
Florian Fr\"owis, Roman Schmied, and Nicolas Gisin.
\newblock ``Tighter quantum uncertainty relations following from a general
  probabilistic bound''.
\newblock \href{https://dx.doi.org/10.1103/PhysRevA.92.012102}{Phys. Rev. A
  {\bf 92}, 012102}~(2015).

\bibitem{min_V_Yu}
Sixia {Yu}.
\newblock ``{Quantum Fisher Information as the Convex Roof of
  Variance}''~(2013).
\newblock  \href{http://arxiv.org/abs/1302.5311}{arXiv:1302.5311}.

\bibitem{Nielsen2002quantum_computation}
Michael~A Nielsen and Isaac Chuang.
\newblock ``Quantum computation and quantum information''.
\newblock Cambridge University Press. ~(2000).

\bibitem{Watrous05}
John Watrous.
\newblock ``Notes on super-operator norms induced by schatten norms''.
\newblock \href{https://dx.doi.org/10.5555/2011608.2011614}{Quantum Info.
  Comput. {\bf 5}, 58--68}~(2005).

\end{thebibliography}

\end{document}